\documentclass[final,11pt]{article}

\usepackage{nicefrac}

\newcommand{\nfrac}{\nicefrac}

\def\FullBox{\hbox{\vrule width 6pt height 6pt depth 0pt}}

\def\qed{\ifmmode\qquad\FullBox\else{\unskip\nobreak\hfil
\penalty50\hskip1em\null\nobreak\hfil\FullBox
\parfillskip=0pt\finalhyphendemerits=0\endgraf}\fi}

\newenvironment{proof}{\begin{trivlist} \item {\bf Proof:~~}}
   {\qed\end{trivlist}}

\usepackage{comment,amsmath,amssymb,amsfonts} 
\usepackage{epsfig} \usepackage{latexsym,nicefrac,bbm}
\usepackage{xspace}
\usepackage{color,fancybox,graphicx,url,subfigure,fullpage}
\usepackage[top=1in, bottom=1in, left=1in, right=1in]{geometry}
\usepackage{tabularx} \usepackage{hyperref}
\usepackage{mathptmx}

\usepackage{enumitem}

\usepackage{booktabs}

\newtheorem{theorem}{Theorem}[section]
\newtheorem{claim}[theorem]{Claim}

\newtheorem{definition}[theorem]{Definition}

\newtheorem{lemma}[theorem]{Lemma}
\newtheorem{corollary}[theorem]{Corollary}

\newtheorem{remark}[theorem]{Remark}
\newcommand{\shrinkspace}[1]{\ifbool{fullversion}{}{\vspace{#1}}}

\newcommand{\defeq}{\stackrel{\textup{def}}{=}}
\newcommand{\eps}{\varepsilon}

\newcommand{\alert}[1]{\ifbool{havealerts}{\marginnote{\textcolor{red}{#1}}}{}}

\newcommand{\C}[1]{\ensuremath{\mathcal{C}\inp{#1}}}

\newcommand{\insq}[1]{\ensuremath{\left[#1\right]}}
\newcommand{\M}[0]{\ensuremath{\mathcal{M}}}

\newcommand{\R}[2][]{\ensuremath{R_{#2}^{#1}}}

\renewcommand{\Pr}[2]{\ensuremath{\mathbb{P}_{#1}\insq{#2}}}

\renewcommand{\C}[0]{\ensuremath{\mathcal{Z}}}

\renewcommand{\vec}[1]{#1}
\renewcommand{\epsilon}[0]{\varepsilon}

\newcommand{\vlambda}{{\vec{\lambda}}}

\newcommand{\vtheta}{{\vec{\theta}}}
\newcommand{\vx}{{\vec{x}}}
\newcommand{\vy}{{\vec{y}}}
\newcommand{\va}{{\vec{a}}}
\newcommand{\iprod}[2]{{\langle {#1}, {#2} \rangle}}

\usepackage[notcite]{showkeys}
\definecolor{labelkey}{gray}{.75}

\renewcommand{\M}{{\mathcal{M}}}

\newcommand{\st}{\mbox{\rm s.t. }}

\def \RR   {{\mathbb R}}

\def \M   {{\mathcal M}}

\def \C   {{\mathcal C}}
\def \B   {{\mathcal B}}
\def \ZZ   {{\mathbb Z}}

\title{\bf Entropy, Optimization and Counting}

\author{Mohit Singh \\  {\small Microsoft Research, Redmond, USA} \\ {\small Email: mohits@microsoft.com} \and Nisheeth K. Vishnoi \\ {\small Microsoft Research, Bangalore, India} \\ {\small Email: nisheeth.vishnoi@gmail.com}}

\date{}

\begin{document}

\maketitle

\thispagestyle{empty}

\newcommand{\one}{{\vec{1}}}

\newcommand{\Authornote}[2]{{\sf\small{\color{red}[#1: #2]}}}

\begin{abstract}
In this paper we study the problem of computing max-entropy distributions over a discrete set of objects subject to observed marginals.
Interest in such distributions arises due to their applicability in areas such as statistical physics, economics, biology, information theory, machine learning, combinatorics and, more recently, approximation algorithms.
A key difficulty in computing max-entropy distributions has been to show that they have polynomially-sized descriptions. 
We show that such descriptions exist under general conditions.
Subsequently, we show how algorithms for (approximately) counting the underlying discrete set can be translated into efficient algorithms to (approximately) compute max-entropy distributions.
In the reverse direction, we show how access to algorithms that compute max-entropy distributions can be used to count, which establishes an equivalence between counting and computing max-entropy distributions.

\end{abstract}

\newpage 
\tableofcontents 

\newpage
\vspace{3mm}

\section{Introduction}
In this paper we study the computability of  max-entropy probability distributions over a discrete set. Consider a collection $\M$ of discrete objects whose building blocks are the elements $[m]=\{1,2,\ldots,m\};$ thus, $\M \subseteq \{0,1\}^m.$
Suppose there is some unknown distribution $p$ on $\M$ and we are given access to it via observables $\vtheta$, the simplest of which is the probability that an element is present in a random sample $M$ from $p$; namely,
$
 \Pr{M \leftarrow p}{e   \in M}=\theta_e.
$
If  $\vec{\theta}$ is all we know, what is our best guess for $p$?
The principle of max-entropy \cite{Jaynes1,Jaynes2} postulates that the best guess is the distribution which maximizes (Shannon) entropy.\footnote{Recall that the Shannon entropy of a distribution $p=(p_M)_{M \in \M}$ is
$H(p) \defeq \sum_{M \in \M} p_M \ln \frac{1}{p_M}.$}
Roughly, the argument is that any distribution which has more information must violate some observable,  and a distribution with less information must implicitly use additional independent observables, hence contradicting their maximality.
Access to such a distribution could then be used to  obtain samples which conform with the observed statistics and to obtain the most informed guess to further statistics.

Given the fundamental nature of such a distribution, it should not be surprising that it shows  up in various areas such as statistical physics, economics, biology, information theory, machine learning, combinatorics and, more recently, in the design of approximation algorithms, see for instance \cite{Kapur}.    From a computational point of view, the  question is how to find max-entropy distributions. Note that the entropy function is concave, hence, the problem of maximizing it over the set of all probability distributions over $\M$  with marginals $\vtheta$ is a convex programming problem.  
But what is the input? If $\vtheta$ and $\M$ are given explicitly, then a solution to this convex program can be obtained, using the 
ellipsoid method, in time polynomial in $|\M|$ and the number of  bits needed to represent $\theta.$\footnote{The ellipsoid algorithm requires a bounding ball which in this case is trivial since $\forall M \in \M,$ $0\leq p_M \leq 1.$}
However, in most interesting applications, while $\theta$ is given explicitly, $\M$ may be  an exponentially-sized set over the universe $\{0,1\}^m$ and 
specified implicitly. 
For example, the input could be a graph $G=(V,E)$ with $m$ edges and $\vtheta \in \mathbb{R}_{\geq 0}^m$ whereas $\M$ could be all spanning trees or all perfect matchings in $G$; in such a scenario $|\M|$ could be exponential in $m.$ 
This renders the convex program for computing the max-entropy distribution prohibitively large. 
Moreover, simply describing the distribution could require exponential space. 
The good news is that one can use convex programming duality to convert the max-entropy program into one that has $m$ variables. 
Additionally, under mild conditions on $\vtheta,$ strong duality holds and, hence, the max-entropy distribution is a product distribution, i.e., there exist $\gamma_e$ for $e \in [m]$ such that for all $M \in \M$, $p_M \propto \prod_{e \in M} \gamma_e,$ see \cite{BoydV04} or Lemma \ref{prop:maxentropya}.
Thus, the max-entropy distribution for $\vec{\theta}$ can be described by $m$ numbers $\vec{\gamma} = (\gamma_e)_{e \in [m]}.$ 
There are two main computational  problems concerning max-entropy distributions. The first is to,  given $\vec{\theta}$ and implicit access to $\M$,  obtain $\tilde{\vec{\gamma}}$ such that the  entropy of the product distribution $\tilde{p},$ corresponding  to $\tilde{\vec{\gamma}},$ is close to that of the max-entropy distribution, and the observables obtained from $\tilde{p}$  are close to $\vec{\theta}.$
The second is, given $\tilde{\gamma},$ obtain a random sample  from the distribution $\tilde{p}.$  The second problem can be handled by invoking the equivalence between approximate counting and sampling due to Jerrum, Valiant and Vazirani \cite{JerrumVV86} and, hence, we focus on the first issue of computing approximations to the max-entropy distributions.\footnote{To be precise, this equivalence between random sampling and approximate counting holds when the combinatorial problem at hand is {\em self-reducible}, see also \cite{JerrumSinclair}.}  
However, the existence of  $\tilde{\gamma}$ which requires polynomially-many  bits in the input size is, a priori,  far from clear.
This raises the crucial question of whether good enough succinct descriptions exist for max-entropy distributions.

While there is a vast amount of literature concerning the computation of max-entropy distributions (see for example the survey \cite{WainwrightJ08}), previous (partial) results on computing max-entropy distributions required exploiting some special structure of the particular problem at hand.
In theoretical computer science,  interest in rigorously computing max-entropy distributions  derives from their applications to randomized rounding and the design of non-trivial approximation algorithms, notably to problems such as the symmetric and the asymmetric traveling salesman problem (TSP/ATSP).  For example, using a very technical argument, \cite{AsadpourGMGS10} give an algorithm to compute the max-entropy distribution over spanning trees of a graph. This algorithm was then used by them to improve the approximation ratio for ATSP, and by \cite{GharanSS11} to improve the approximation ratio for (graphical) TSP, making progress on two long-standing problems.  
Subsequently, the ability to compute max-entropy distributions over spanning trees has also been used  to design efficient privacy preserving mechanisms for spanning tree auctions by \cite{HuangK12}. 
In another example, \cite{AsadpourS10} show how to compute max-entropy distributions over perfect matchings in a tree and use it to design approximation algorithms for a max-min fair allocation problem.  The question of computing max-entropy distributions over perfect matchings in bipartite (and general) graphs, however,  has been an important open problem. Recent applications of the ability to compute max-entropy distribution over perfect matchings in bipartite graphs  include new approaches for  TSP and ATSP, see \cite{Vishnoi12} and Section \ref{sec:tsp}.

For counting problems, it is rare to obtain algorithms that can count exactly, notable exceptions being the problem of counting spanning trees in a graph~\cite{Kirchoff47} or counting certain problems restricted to trees using dynamic programming. Most natural counting problems turn out to be $\#$P-hard including the problem to count the number of perfect matchings in a bipartite graph~\cite{Valiant79a,Valiant79b}. The goal then shifts to finding algorithms that approximately count up to any fixed precision~\cite{KarpL83,StockMeyer85}. Here the most successful technique has been the Markov chain Monte Carlo (MCMC) method~\cite{JerrumS89} which, when combined with the equivalence between approximate counting and sampling~\cite{JerrumVV86}\footnote{We note that MCMC methods efficiently sample (and count) given a fixed $\vec{\gamma},$ usually for $\vec{\gamma}=\one$ corresponding to the uniform distribution. The goal in our problem is to find an  ${\vec{\gamma}}$ that maximizes the entropy. In fact, given a ${\vec{\gamma}},$  problem specific MCMC methods can be used to generate a random sample from $\M$ according to the product distribution corresponding to $\gamma.$}, leads to many approximate counting algorithms. The technique has been applied to many problems including counting perfect matchings in a bipartite graph~\cite{JerrumSV04}, counting bases in a balanced matroid~\cite{FederM92}, counting solutions to a knapsack problem~\cite{DyerFKKPV93} and counting the number of colorings in restricted graph families~\cite{Jerrum95}. 
However, the problem of obtaining approximate counting oracles for several problems remains open as well; perhaps a prominent example is that of (approximately) counting the number of perfect matchings in a general graph.

In combinatorics, max-entropy distributions are often referred to as  {\em hard-core} or Gibbs distributions and have been intensely studied. While it is nice that the hard-core  distribution has a product form, i.e., $p_M \propto \prod_{e \in M}\gamma_e,$ the question of interest here is whether one can upper bound the $\gamma_e$s. Structurally, such a bound implies  that hard-core distributions exhibit a significant amount of approximate stochastic independence.
In an important result,  \cite{KahnK97} proved such a bound for the hard-core distribution over matchings in a graph.
This led to resolving several questions involving asymptotic graph and hypergraph problems.  For instance results of  \cite{Kahn96,KahnK97, Kahn00} prove that the fractional chromatic index of a graph asymptotically behaves as its chromatic index. However, the argument of \cite{KahnK97} is quite difficult and seems to be specific to the setting of matchings leaving it an interesting problem to understand under what conditions can one obtain upper bounds  on $\gamma_e$s.

\subsection{Our Contribution}
We first show that good enough succinct representations exist for max-entropy distributions. 
Subsequently, we give an algorithm that computes arbitrarily good approximations to max-entropy distributions given $\vtheta$ and access to a suitable counting oracle for $\M.$ Our algorithm is efficient whenever the corresponding counting oracle is efficient. Moreover, the counting oracle can be approximate and/or randomized. This allows us to leverage a variety of algorithms developed for several $\#$P-hard problems to give  algorithms to compute  max-entropy distributions. Consequently, we obtain several new and old results about concrete algorithms to compute max-entropy distributions. Interesting examples  for which we can use pre-existing counting oracles to obtain max-entropy distributions   include  spanning trees, matchings in general graphs, perfect matchings in bipartite graphs (using the  algorithm from \cite{JerrumSV04}) and subtrees of a rooted tree. 
The consequence for perfect matchings in bipartite graphs makes the  algorithmic strategies for TSP/ATSP mentioned earlier computationally feasible, see Section \ref{sec:tsp}. 
In  the reverse direction we show that if one can solve, even approximately, the convex optimization problem of computing the max-entropy distribution, one can obtain such counting oracles. This establishes an equivalence between counting and computing max-entropy distributions in a general setting.
As a corollary,  we obtain that the  problem of  computing max-entropy distributions over perfect matchings in general graphs is equivalent to the, hitherto unrelated, problem of approximately counting perfect matchings in a general graph.

\subsubsection{Informal Statement of Our Results}
Before we describe our results a bit more technically, we introduce some basic notation.
For $\M \subseteq \{0,1\}^m,$ let $P(\M)$ denote the convex hull of
all $\M$ where each $M \in \M$ is thought of as a $0/1$ vector of length $m$, denoted $\vec{1}_M.$ Thus, given a  $\vtheta,$ for the max-entropy program to have any solution, $\vtheta \in P(\M).$  
Since we are concerned with the case when $\M$ is given implicitly and may be of exponential size,  we no longer hope to solve the max-entropy convex program  directly since that may require  exponentially many variables, one  for each $M \in \M.$ Thus, we work with the dual to the max-entropy convex program. 
The dual has $m$ variables and, if $\vtheta$ is in the relative interior of $P(\M),$ the optimal dual solution can be used to describe the optimal solution to the max-entropy convex program, which is a product distribution. In fact, we  assume one can put a ball of radius $\eta$ around $\vtheta$ and it still remains in the interior of $P(\M).$ 
Importantly, our algorithm requires access to a {\em generalized counting oracle} for $\M,$ which given $\vec{\gamma}$ can compute $\sum_{M \in \M, \; M \ni e} \prod_{e' \in M}\gamma_{e'}$ for all $e \in [m]$ and also the sum  $\sum_{M \in \M} \prod_{e \in M}\gamma_e.$ We also consider the case when the oracle is approximate (possibly randomized) and for a given $\eps,$ can output the sums  above up to a multiplicative error of $1 \pm \eps.$   
The following is the first main result of the paper, stated informally here.
\begin{theorem}[Counting Implies Optimization, See Theorems \ref{thm:strong} and \ref{thm:approximate-count}]\label{thm:main1}
There is an algorithm which,
given access to a generalized (approximate) counting oracle for $\M \subseteq \{0,1\}^m,$  a $\vtheta$ which is promised to be in  the $\eta$-interior of $P(\M)$ and an $\eps>0,$   outputs a $\vec{\gamma}$ such that its corresponding product probability distribution $p$ is such that $H(p) \geq (1-\nfrac{\eps}{\eta}) H(p^\star)$ and for every $e \in [m],$
$$ \left| \Pr{M \leftarrow p}{e \in M} - \theta_e \right|  \leq \eps.$$
Here, $p^\star$ is the max-entropy distribution corresponding to $\vtheta.$
The number of calls the algorithm makes to the oracle is bounded by a polynomial  in the input size, $\ln \nfrac{1}{\eta}$ and $\ln \nfrac{1}{\eps}.$
\end{theorem}

\noindent
A useful setting for $\eta$ and $\eps$ to keep in mind is $\nfrac{1}{m^2}$ and $\nfrac{1}{m^3}$ respectively.
The bit-lengths of the inputs to the counting oracle are  polynomial in $\nfrac{1}{\eta}$ and, hence,  the running time of our algorithm depends polynomially on $\nfrac{1}{\eta}.$
If the generalized counting oracle is $\eps$-approximate, the same guarantee holds. Note that for many approximate counting oracles, the dependence on $\eps$ on their running time is a polynomial in $\nfrac{1}{\eps}.$  Hence, in this case the running time depends polynomially on $\nfrac{1}{\eps}.$
Finally, note that this result can be easily generalized to obtain algorithms for the  problem of finding the distribution that minimizes the Kullback-Leibler divergence from a given product distribution subject to the marginal constraints, see Remark \ref{rem:gencount}.

At a very high level, the algorithm in this theorem is obtained by applying the framework of the ellipsoid algorithm to the dual of the max-entropy convex program. While it is more convenient to work with the dual  since it has $m$ variables  two issues arise: The domain of optimization becomes unconstrained and the separation oracle requires the ability to compute (possibly exponential sums) over subsets of $\M.$ While the counting oracles can be adapted to compute  exponential sums, the unboundedness of the domain of optimization is an important problem.

One of the  technical results in the proof of the theorem above is structural and shows that this dual optimization problem has an optimal solution in a box of size 
$\nfrac{m}{\eta}$ when $\vtheta$ is in the $\eta$-interior of $P(\M),$ see Theorem \ref{thm:radius}.  Since $\gamma_e$s are exponential in the respective dual variables,  there is an approximation $\tilde{\gamma}$ to the optimal solution to the max-entropy program, when $\theta$ is in the $\eta$-interior of $P(\M),$ such that the number of bits needed to represent each $\tilde{\gamma}_e$ is at most $\nfrac{m}{\eta}.$  
Such a result has been obtained for the special case of spanning trees by \cite{AsadpourGMGS10} and for matchings in a general graph by \cite{KahnK97}.

Given that counting algorithms for many problems are still elusive, one may ask if they are really necessary to compute max-entropy distributions. The final result of this paper answers this question in the affirmative and establishes  a converse to Theorem \ref{thm:main1}. 

\begin{theorem}[Optimizing Implies Counting, see Theorem \ref{thm:reverse}]
There is an algorithm, which given oracle access to an algorithm to compute an $\eps$-approximation to the max-entropy convex program for an $\eta$-interior point of $P(\M),$ and a separation oracle for $P(\M),$  can  compute a number $Z$ such that
$$ (1-\eps) |\M|  \leq Z \leq (1+\eps) |\M|.$$
The number of calls made to the max-entropy oracle is  polynomial in the input size and $\nfrac{1}{\eps}.$
\end{theorem}
\noindent
This result can be extended to obtain generalized counting oracles, see Remark \ref{rem:gencount}.
For all polytopes of interest in this paper, separation oracles are known, see Section \ref{sec:examples}. Moreover, this result continues to hold even when the separation oracle is approximate, or {\em weak}. 
As a corollary, using a separation oracle for the perfect matching polytope for general graphs \cite{Edmonds65b,padberg1982odd}, we obtain that an algorithm to compute a good-enough approximation to the max-entropy distribution for any $\vtheta$ in the  perfect matching polytope of a graph $G$ implies an FPRAS to count the  number of perfect matchings in the same graph.

\subsection{Technical Overview}
The starting point for our results is the following dual to the max-entropy convex program:
\begin{equation}\label{eq0-3}
\inf _{\vlambda} \iprod{\vlambda}{\vtheta} + \ln \sum_{M \in \M} e^{-\iprod{\vlambda}{\vec{1}_M}},
\end{equation}
where $\vec{1}_M$ is the indicator vector for $M.$
When $\vtheta$ lies in the relative interior of $P(\M),$ then strong duality holds between the primal and the dual.\footnote{When $\theta$ lies on the boundary of $P(\M),$ the infimum in the dual is not attained for any finite $\lambda.$}
Hence,  it follows from the first order conditions on the optimal solution pair $(p^\star,\vlambda^\star)$ that
$p^\star_M  \propto e^{-\iprod{\vlambda^\star}{\vec{1}_M}}$ for each $M\in \M$.
Suppose we know that
\begin{enumerate}
\item   $\vlambda^\star$ is  bounded, i.e.,  $\|\vlambda^\star\| \leq R$ for some $R,$ and
\item  there is a  generalized counting  oracle that allows us to compute the gradient of the objective function $f(\vlambda) \defeq  \iprod{\vlambda}{\vtheta} + \ln \sum_{M \in \M} e^{-\iprod{\vlambda}{\vec{1}_M}}$ at a specified $\vlambda.$  The gradient at $\vlambda,$  denoted $\nabla f(\vlambda),$ turns out to be a vector whose coordinate corresponding to $e \in [m]$ is
$$ \theta_e -\frac{ \sum_{M \in \M, \; M \ni e} e^{-\iprod{\vlambda}{\vec{1}_M}}}{ \sum_{M \in \M} e^{-\iprod{\vlambda}{\vec{1}_M}}}.$$
\end{enumerate}
Then, using the machinery of the ellipsoid method, it follows relatively straight-forwardly that, for any $\eps,$ we can compute a point $\vlambda^\circ$ such that
$f(\vlambda^\circ) \leq f(\vlambda^\star)+\eps$ with at most a {\rm poly}$(m,\ln \nfrac{1}{\eps}, \ln R)$  calls to the counting oracle. Note that since the numbers fed into the counting oracle are  of the form $e^{-\lambda_e},$ for each $e \in [m],$ the running time of the counting oracle depends polynomially on $R$ rather than $\ln R.$ Thus, we  need $R$ to be polynomially bounded. Hence,  the question is:
\begin{center}Can we bound $\|\vlambda^\star\|?$
\end{center}
Indeed, a significant part of the work done in  \cite{AsadpourGMGS10} was to bound this quantity for  spanning trees and in \cite{KahnK97} for  (not necessarily perfect)  matchings in   graphs.
A priori it is not clear why there should be any such bound. In fact, we observe that if $P(\M)$ lies in some low-dimensional affine space in $\mathbb{R}^m,$   the optimal solution is not unique and can be shifted in any direction normal to the space, see Lemma \ref{lem:shift0}. Thus, one can only hope for the optimal solution to be bounded once one imposes the restriction that $\vlambda^\star$ lies in the linear space corresponding to the affine space in which $P(\M)$ lives.
One thing that works in our favor is that we have an absolute upper bound (independent of $\vtheta$) on the optimal value of $f( \cdot),$ namely $m.$ Roughly, this is because at optimality this quantity is an entropy over a discrete set of size at most $2^m.$
This  implies that for all $M \in \M,$
\begin{equation}\label{eq0-4} \iprod{\vlambda^\star}{\vtheta} -  \iprod{\vlambda^\star}{\vec{1}_M} \leq m.
\end{equation}
Using this,  the $\eta$-interiority of $\vtheta,$ and the fact that the diameter of $P(\M)$ is at most $\sqrt{m},$   it can then be shown that
$$ \max_{M \in \M} \iprod{\vlambda^\star}{\vec{1}_M} -  \min_{M \in \M} \iprod{\vlambda^\star}{\vec{1}_M} \leq \frac{m \sqrt{m}}{\eta}.$$
Let us show how this immediately implies a bound on $R$ when $\M$ corresponds to all the spanning trees of a graph with no bridge. Suppose $T$ and $T'$ are two trees such that $T'$ is obtained from $T$ by deleting an edge $e$ and adding an edge $f,$ then,
$$ |\iprod{\vlambda^\star}{\vec{1}_T} - \iprod{\vlambda^\star}{\vec{1}_{T'}}| =|\lambda^\star_e-\lambda^\star_f| \leq \frac{m \sqrt{m}}{\eta}.$$ Thus, unless the graph has a bridge, this implies  $|\lambda^\star_e-\lambda^\star_f| \leq \frac{m^2 \sqrt{m}}{\eta}$ for all $e,f \in G.$ 
However, attempting a similar combinatorial argument for perfect matchings in a bipartite graph, where we do not have this {\em exchange} property, the bound  is worse by a factor of $2^m.$

Thus, we abandon combinatorial approaches and appeal to the geometric implication  of \eqref{eq0-4} to obtain the desired polynomial bounding box for all $P(\M).$ The argument is surprisingly simple and we sketch it here. One way to interpret \eqref{eq0-4} is that the vector $\vlambda' \defeq -\nfrac{\vlambda^\star}{m}$ has inner product at most $1$ with  $\vec{v}-\vtheta$ for all $\vec{v} \in P(\M).$ For now, neglecting the fact that $P(\M)$ might live in a lower dimensional affine space and that $
\vec{0}$ may not be in $P(\M),$ this implies that $\vlambda'$ is in the {\em polar} of $P(\M).$ However, since $\vtheta$ is in the $\eta$-interior of $P(\M),$ $P(\M)$ contains an $\ell_2$ ball of radius at least $\eta$ inside it. Thus, the polar of $P(\M)$ must be contained in the polar of this ball, which is nothing but an $\ell_2$ ball of radius $\nfrac{1}{\eta}.$ This gives a bound of $\nfrac{1}{\eta}$ on the $\ell_2$ norm of  $\vlambda'$ and, hence, a bound of $\nfrac{m}{\eta}$ on the norm of $\vlambda^\star$ as desired.

Thus, the ellipsoid method can be used to obtain a solution $\vlambda^\circ$ such that $f(\vlambda^\circ) \leq f(\vlambda^\star) +\eps.$ Why should this approximate bound imply that the product distribution obtained using $\vlambda$ is close  in the marginals to $\vtheta$? The observation here is that $f(\vlambda^\circ) - f(\vlambda^\star)$ is the Kullback-Leibler (KL) divergence between the two distributions. This  implies a  bound of $\sqrt{\eps}$ on the marginals using a standard upper bound  on the total variation distance in terms of the KL-divergence.

In  the case when we have access only to  an approximate counting oracle for $\M,$ things are  more complicated. Roughly, the approximate counting oracle translates to having access to an approximate  gradient oracle for the function $f(\cdot)$ and one has to ensure that $\vlambda^\star$ is not cut-off during an iteration. Technically, we  show that this does not happen and, hence, approximate counting oracles are equally useful for obtaining good approximations to max-entropy distributions. 

Finally, note that the (projected-)gradient descent approach (see  \cite{Nesterov}) can also be shown to converge in polynomial time and, possibly, can result in practical algorithms for computing max-entropy distributions. In the case when the counting oracle is approximate, one has to deal with a {\em noisy} gradient and  the solution turns out to be similar to the one  in the ellipsoid method-based algorithm in the presence of an approximate counting oracle.
In addition to a bound on $\|\lambda^\star\|,$ one needs to   bound  the $2 \rightarrow 2$ norm of the gradient of $f.$ While we omit the details of the gradient descent based-algorithm, we  show that $\|\nabla f\|_{2 \rightarrow 2}$ is polynomially bounded, see Remark \ref{rem:grad} and Theorem \ref{thm:2to2} in Appendix \ref{sec:grad}. This bound may be of independent interest. 

We now give an overview of the reverse direction: How to count approximately given the ability to solve the max-entropy convex program for any point $\vtheta$ in the $\eta$-interior of $P(\M).$ We start by noting that if we consider  $\vtheta^\star \defeq \frac{1}{|\M|} \sum_{M \in \M} \vec{1}_M,$ then  the optimal value of the convex program is $\ln |\M|.$  Thus, given access to this vertex-centroid of $P(\M)$ one can get an estimate of $|\M|.$ However, computing $\theta^\star$ can be shown to be as hard as counting $|\M|,$ for instance, when $\M$ consists of  perfect matchings in a bipartite graph, see \cite{Elbassioni201256}. 
We bypass this obstacle and  apply the ellipsoid algorithm on the following (convex-programming) problem
$$\sup_{\vtheta} \inf_{\vlambda} f_{\vtheta}(\vlambda)$$ where $f_{\vtheta}(\vlambda)$ is the function in \eqref{eq0-3} and where we have chosen to highlight the dependence on $\vtheta.$ The ellipsoid algorithm proposes a $\theta$ and expects the max-entropy oracle to output an approximate value for $ \inf_{\vlambda} f_{\vtheta}(\vlambda).$ 
This raises a few  issues: 
First, given our result on optimization via counting, it is unfair to assume that we have such an oracle that works for all $\vtheta,$ irrespective of the interiority of $\theta$ in  $P(\M).$ Thus, we allow queries to the oracle {\em only} when $\vtheta$ is sufficiently in the interior of $P(\M).$ 
Note that our algorithm for computing the max-entropy distribution in our first theorem works under  these guarantees. This requires, in addition, a separation oracle for checking whether a point is in the $\eta$-interior of $P(\M).$ 
We construct such an $\eta$-separation oracle from a  separation oracle for $P(\M).$ The latter,   given a point, either says it  is in $P(\M)$ or returns an inequality valid for $P(\M)$ but violated by this point.

The second issue is that $\theta^\star,$ our target point, may not be in the $\eta$-interior of $P(\M).$ In fact, there may not be {\em any} point in the $\eta$-interior of $P(\M)$ when $\eta$ is $\nfrac{1}{{\rm poly}(m)}.$  However, under reasonable conditions on $P(\M),$ which are satisfied for all polytopes we are interested in, we can show that there is a point $\theta^\bullet$ in the $\eta$-interior of $P(\M).$  This  allows us to recover a good enough estimate of $|\M|.$ Thus (the way we apply the framework of ellipsoid algorithm), we are able to  recover a  point close enough to $\theta^\bullet$ by doing a binary search on the target value of $|\M|.$
As in the forward direction, because we assume that the max-entropy algorithm is approximate, we must argue that $\theta^\bullet$ is not cut-off during any iteration of the ellipsoid algorithm. 

We conclude this overview with a couple of remarks. First, unlike our results in the forward direction, we cannot replace the ellipsoid method based algorithm by a gradient descent approach. The reason is that we only have a separation oracle to detect whether a point is in  $P(\M)$ or not.
Second, we can extend our result to show that, using a max-entropy oracle, one can obtain {\em generalized} approximate counting oracles, see Remark \ref{rem:gencount}.  

\subsection{Organization of the Rest of the Paper}
The rest of the paper is organized as follows. In Section~\ref{sec:prelim}, we formally define the objects of interest in our paper including the convex program for optimizing the max-entropy distribution and its dual. We also define counting oracles that are needed for solving the convex program. We then formally state our results and give a few lemmas stating properties about the optimal and near-optimal solution to the dual of the max-entropy convex program.
In Section~\ref{sec:examples} we  provide examples of some combinatorial polytopes to which our results apply. In Section \ref{sec:tsp}, we show how certain algorithmic approaches for approximating the symmetric and the asymmetric traveling salesman problem are feasible as a result of one of the main results of this paper.
In Section~\ref{sec:bounding}, we prove that there is an optimal solution to the dual of the max-entropy convex program that is contained in a ball of small radius around the origin. In Section~\ref{sec:algorithm}, we use this bound on the optimal solution to show that counting oracles, both exact and approximate, can be used to optimize the convex program via the ellipsoid algorithm. In Section~\ref{sec:reverse}, we show the other direction of the reduction and give an algorithm that can approximately count given an oracle that can approximately solve the max-entropy convex program.
Standard proofs are omitted from the main body and appear in Appendix \ref{sec:omit}. In Appendix \ref{sec:general} we show how generalized counting oracles can be obtained via max-entropy oracles. Here, we also introduce the program for minimizing the KL-divergence with respect to a fixed distribution.
Finally, in Appendix \ref{sec:grad}, we give a bound on $\|\nabla f\|_{2 \rightarrow 2}.$

\section{Preliminaries}\label{sec:prelim}

\subsection{Notation} In this section we introduce the general notation used throughout the paper.
Vectors are denoted by plain letters such as $a,b,c,d,x,y,u$ and $v$ and are over $\mathbb{R}^m.$ We  also use the Greek letters $\lambda,\theta, \nu$ and $\gamma$ to denote vectors. ${0}$ is  sometimes used to denote the all-zero vector and the usage should be clear from  context.
 For reasons emanating from applications, we  choose to index the set $[m]$ by $e.$ Hence, the components of a vector are denoted by $x_e,\lambda_e, \theta_e,$ etc. We also use notation such as  $x_0,x_1,\ldots, x_t$ and $\lambda_0,\lambda_1,\ldots,\lambda_t$  to denote vectors. It should be clear from the context that these are vectors and not their components.
  The Greek letters  $\eta,\alpha, \beta, \eps,\zeta$  are used to denote positive real numbers.
For a set $M \in \{0,1\}^m,$ let $1_M$  denote the $0/1$ indicator vector for $M.$ We  use $1_M(e)$ to denote its $e$-th component. Thus, $1_M(e)=1$ if $e \in M$ and $0$ otherwise.
The letters $p,q$ and $r$ are reserved to denote probability distributions over $\{0,1\}^m.$ Of special interest are   product probability distributions where, for $M \in \{0,1\}^m,$  the probability of $M$ is proportional to $\prod_{e \in M} \gamma_e$ for some vector $\gamma.$ We  denote such a probability distribution by $p^\gamma$ to emphasize its dependence on $\gamma,$ and let $p^\gamma_M$ denote the probability of $M.$  Additionally $\langle x,y \rangle$ denotes the inner product of two vectors, $\| \cdot \|$ denotes the Euclidean norm and $\| x \|_\infty \defeq \max_{e \in [m]} |x_e|.$ We also use the notation $\lambda(M)$ to denote $\langle \lambda, 1_M \rangle$  for a vector $\lambda$ and $M \in \{0,1\}^m.$ $|S|$ denotes the cardinality of a set.

\subsection{Combinatorial Polytopes, Separation Oracles, Counting Oracles and Interiority}

The polytopes of interest arise as convex hulls of subsets of $\{0,1\}^m$ for some $m.$
For a set $\M \subseteq \{0,1\}^m,$ the corresponding polytope is denoted by $P(\M).$  Thus,
$$ P(\M) \defeq \left\{ \sum_{M \in \M} p_M \one_M: p_M \geq 0,\ \ \sum_{M \in \M} p_M =1 \right\}.$$
Another way to describe $P(\M)$ is to give a maximal set of linearly independent equalities satisfied by all its vertices, and to list the inequalities that define $P(\M).$ Thus, $P(\M)$ can be described by  $(A^=, \vec{b})$ and $(A^\leq, \vec{c})$  such that
$$\forall_{M \in \M} \ \quad A^=\one_M=\vec{b} \ \quad \mbox{and} \ \quad A^\leq \one_M \leq \vec{c}.$$
While the former set cannot be more than $m,$ the latter set can be exponential in $m$ and we do not assume that $(A^\leq,c)$  is given   to us explicitly.

\vspace{2mm}
\noindent
{\bf Separation Oracles.} On occasion we  require an access to a {\em separation oracle} for $P(\M)$ of the following form: Given $\vec{\lambda} \in \mathbb{R}^m$ satisfying $A^=\vec{\lambda} = \vec{b},$ the separation oracle either says that $A^\leq \vec{\lambda} \leq \vec{c}$ or  outputs an inequality $(a',c')$  such that
$$\langle a', \vec{\lambda} \rangle > c'.$$
In fact, such an oracle is often termed  a {\em strong} separation oracle.\footnote{In our results that depend on access to a strong separation oracle, we can relax the guarantee to that of a {\em weak}  separation oracle. We omit the details.}

\vspace{2mm}
\noindent
{\bf Counting Oracles.} The standard counting problem associated to $\M$ is to determine $|\M|,$ i.e.,  the number of vertices of $P(\M).$
We are interested in a more general counting problem associated to $\M$ where there is a weight $\lambda_e$ for each $e \in [m]$ and the weight of  $M$ under this measure is  $e^{-\lambda(M)}.$ A {\em generalized exact counting oracle} for $\M$ then outputs the following two quantities:
\begin{enumerate}
\item $Z^{\vlambda} \defeq \sum_{M \in \M} e^{-\lambda(M)} $ and
\item for every $e \in [m],\;\; Z^{\vlambda}_e \defeq \sum_{M \in \M,\; M \ni e} e^{-\lambda(M)}.$
\end{enumerate}
The oracle is assumed to be efficient, as in it runs in time polynomial in $m$ and bits needed to represent $e^{-\lambda_e}$ for any $e\in [m]$.\footnote{To deal with issues of irrationality, it  suffices to obtain the first $k$ bits of  $Z^{\vlambda}$ and $Z^{\vlambda}_e$ in time polynomial in $k$ and $m$.}

While efficient generalized exact counting oracles are known for some settings, for many problems of interest the exact counting problem is \#P-hard. However, often, for these \#P-hard problems, efficient oracles which can compute arbitrarily good approximations to the quantities of interest are known. Thus, we have to relax the notion to {\em generalized approximate counting oracles} which are possibly randomized. Such an oracle,  given $\epsilon,\alpha>0$ and  weights  $\vec{\lambda}\in \mathbb{R}^{m}$,  returns $\widetilde{Z}^{{\vlambda}}$ and $\widetilde{Z}^{\vlambda}_e$ for each $e\in [m].$ The following guarantees hold with probability at least $1-\alpha$,
\begin{enumerate}
\item $(1-\epsilon) Z^\lambda \leq \widetilde{Z}^{{\vlambda}} \leq (1+\epsilon)  Z^\lambda$ and
\item for every $e \in [m],$
$(1-\eps) Z^\lambda_e \leq \widetilde{Z}^{\vlambda}_e \leq (1+\eps) Z^\lambda_e .$
\end{enumerate}
The running time is polynomial in $m$,  $\nfrac{1}{\epsilon},$  $\log {\nfrac{1}{\alpha}}$ and the number of bits needed to represent $e^{-\lambda_e}$ for any $e\in [m]$. For the sake of readability, we  ignore the fact that approximate counting oracle may be randomized. The statements of the theorems that use randomized approximate counting oracles can be modified appropriately to include the dependence on $\alpha.$
Note that if the problem at hand is self-reducible, then having access to an oracle that outputs an approximation to just  $Z^\lambda$ suffices. We omit the details and the reader is referred to a discussion on self-reducibility and  counting in \cite{JerrumSinclair}.
Finally, it can be shown that, in our setting, the existence of a generalized (exact or approximate) counting oracle is a stronger requirement than the existence of a separation oracle.

\vspace{2mm}
\noindent
{\bf Interior of the Polytope.} The dimension of $P(\M)$ is $m- {\rm rank}(A^=);$  the polytope restricted to this affine space is full dimensional.
Since we  work with polytopes that are not full dimensional, we extend the notion of  the interior of the polytope $P(\M)$ and use the following definition.

\begin{definition}
For an $\eta>0,$ a point $\vec{\theta}$ is said to be in the $\eta$-interior of $P(\M)$ if
$$\left\{\vec{\theta}': A^=\vec{\theta}'=\vec{b} , \  \|\vec{\theta}-\vec{\theta}'\|\leq \eta \right\}\subseteq P(\M).$$
We say that $\vec{\theta}$ is in the interior of $P(\M)$ if $\vec{\theta}$ is in the $\eta$-interior of $P(\M)$ for some $\eta>0.$
\end{definition}

\noindent
We are be interested in the case where $\eta \geq \frac{1}{{\rm poly}(m)}.$ Hence, it is natural to ask if for every $P(\M),$ there is a point in its $\frac{1}{{\rm poly}(m)}$-interior. The following lemma, whose proof appears in Section \ref{sec:reverse}, asserts that the answer  is yes if the entries of $A^\leq$ and $c$ are {\em reasonable} (as is the case in all our applications).

\begin{lemma}\label{lem:interiority0}
Let $\M \subseteq \{0,1\}^m$ and  $P(\M)=\left\{\vec{x}\in \RR^m_{\geq 0}: A^= \vec{x}=\vec{b}, \; A^{\leq}\vec{x}\leq \vec{c}\right\}$ be such that all the entries in $A^{\leq},\vec{c}\in \frac{1}{{\rm poly}(m)} \cdot \ZZ$ and their  absolute values are at most ${\rm poly}(m).$
Then there exists a ${\tilde{\theta}}\in P(\M)$ such that $\tilde{\theta}$ is in the $\frac{1}{{\rm poly}(m)}$-interior of $P(\M).$
\end{lemma}

\noindent
At this point, if one wishes, one can look at Section \ref{sec:examples} for some examples of combinatorial polytopes we consider in this paper.

\subsection{The Maximum Entropy Convex Program}
In this section we present the  convex program for computing  max-entropy distributions. Let $P(\M)$ be the polytope corresponding to $\M.$ While we do not care about whether we have an oracle for $\M$  in this section,  the notion of interiority is important.

For any point $\vec{\theta} \in P(\M),$ by definition, it can be written as a convex combination of vertices of $P(\M),$ each of which is indicator vector for some $M\in \M$. Each such convex combination is  a probability distribution over $M \in \M.$  Of central interest in this paper is a way to find the convex combination
that maximizes the entropy of the underlying probability distribution.
Given $\vec{\theta},$ we can express the problem of finding the max-entropy distribution over the vertices of $P(\M)$ as the  program in Figure \ref{conv-entropya}.
\begin{figure}[h]
\begin{center}
\begin{eqnarray*}
  \sup & \sum_{M\in \M} p_M \ln\frac{1}{p_M}  \\
\st  \\
\forall e\in [m]& \sum_{M\in \M, \; M \ni e} p_M=\theta_e \\
  & \sum_{M\in \M} p_M =1\\
\forall M\in \M & p_M \geq 0
\end{eqnarray*}
\end{center}
\caption{Max-Entropy Program for ($\M,\vec{\theta}$)} \label{conv-entropya}
\end{figure}
Here $0 \ln \frac{1}{0}$ is assumed to be $0.$ This entropy function is easily seen to be concave and, hence, maximizing it is a convex programming problem.
The following folklore lemma, whose proof appears in Appendix \ref{sec:omit-duality}, shows that  if $\vec{\theta}$ is in the interior of $P(\M),$ then the max-entropy distribution corresponding to it is unique and can be succinctly represented.
Recall the notation that for  $\lambda:[m] \mapsto \mathbb{R}$ and $M \in \M,$
$\lambda(M) \defeq \langle \vec{\lambda} , \one_M \rangle=  \sum_{e \in M} \lambda_e.$

\begin{lemma}\label{prop:maxentropya}
For a point $\vec{\theta}$ in the interior of $P(\M),$  there exists a unique distribution $p^\star$ which attains the max-entropy while satisfying
$$\sum_{M \in \M} p^\star_M \one_M=\vec{\theta}.$$ Moreover, there exists a $\lambda^\star :[m] \mapsto \mathbb{R}$ such that $p^\star_M\propto e^{-\lambda^\star(M)}$ for each $M\in \M$.
\end{lemma}

\noindent
As we observe soon,  while $p^\star$ is unique, $\lambda^\star$ may not be.
First,  we record the following definitions about such product distributions.
\begin{definition}
For any $\vec{\lambda}\in \mathbb{R}^{m}$, we define the distribution $p^{\vlambda}$ on $\M$ such that
$$\forall {M \in \M} \quad p_M^{\vlambda} \defeq \frac{ e^{-\lambda(M)}}{Z^{\vlambda}} \textrm{ where } Z^{\vlambda} \defeq \sum_{N\in \M} e^{-\lambda(N)}.$$ The marginals of such a distribution are denoted by $\vec{\theta}^{\vlambda}$ and defined to be
$$\theta^{\vlambda}_e \defeq \sum_{M \in \M, \; M \ni e} p^{\vlambda}_M= \frac{Z^{\vlambda}_e}{Z^{\vlambda}}\textrm{ where } Z^{\vlambda}_e \defeq \sum_{M \in \M, \; M \ni e} e^{-\lambda(M)}.$$
\end{definition}

\noindent
The proof of this lemma  relies on establishing that strong duality holds  for computing the max-entropy distribution with marginals $\vec{\theta}$ for the  convex program in Figure \ref{conv-entropya}. The dual of this program appears in Figure \ref{conv-dual-entropya}.
Thus, if $\theta$ is in the interior of $P(\M),$ then there is a  $\lambda^\star$ such that $p^\star = p^{\lambda^\star}$ and
$f_\theta(\lambda^\star)=H(p^\star).$

Note that
$\lambda^\star$ may not be unique and,
finally,  that an important property of the dual objective function is that  $f_\theta$ does not change if we {\em shift}  by a vector in the span of the rows of $A^=.$ This is  captured in the following lemma whose proof appears in Appendix \ref{sec:omit-near-optimal}.

 \begin{lemma}\label{lem:shift0}
$f_\theta(\vec{\lambda})= f_\theta(\vec{\lambda}+ (A^=)^\top d)$ for any $\vec{d}.$
\end{lemma}

\noindent
Thus, we can restrict our search for the optimal solution to the set
$ \{\vec{\lambda}\in \mathbb{R}^m:  A^=\vec{\lambda}=\vec{0}\}.$  In this set there is a unique $\lambda^\star$ which achieves the optimal value since  the constraints $(A^=,b)$ are assumed to form a maximal linearly independent set. We  refer to this  $\lambda^\star$  as the unique solution to the dual convex program.

\begin{figure}[h]
\begin{center}
\begin{eqnarray*}
\inf & f_\theta(\vec{\lambda})\defeq \sum_{e\in [m]}\theta_e \lambda_e+\ln \sum_{M\in \M}e^{-\lambda(M)} \\
\st \\
\forall e\in [m] & \lambda_e \in \mathbb{R}
\end{eqnarray*}
\end{center}
\caption{Dual of the Max-Entropy Program for ($\M,\vec{\theta}$)} \label{conv-dual-entropya}
\end{figure}

\subsection{Formal Statement of Our Results}\label{sec:formal}

Our first result shows that if one has access to an  generalized exact counting oracle then one can indeed compute a good approximation to the max-entropy distribution for specified marginals.

\begin{theorem}\label{thm:strong}
There exists an algorithm that, given a maximal set of linearly independent equalities $(A^=,b)$ and a  generalized exact counting oracle for $P(\M) \subseteq \mathbb{R}^m,$   a $\vec{\theta}$ in the $\eta$-interior of $P(\M)$ and an $\epsilon>0,$ returns a $\vec{\lambda}^\circ$ such that
$$f_\theta(\vec{\lambda}^\circ) \leq f_\theta(\vec{\lambda}^\star) + \eps,$$
where $\vec{\lambda}^{\star}$ is the optimal solution to the dual of the max-entropy convex program for $(\M,\vec{\theta})$ from Figure \ref{conv-dual-entropya}. Assuming that the generalized exact counting oracle is polynomial in its input parameters, the running time of the algorithm is polynomial in $m$, $\nfrac{1}{\eta},$ $\log \nfrac{1}{\epsilon}$ and the number of bits needed to represent  $\theta$ and  $(A^=,b).$
\end{theorem}

\noindent
The proof of this theorem follows from an application of the ellipsoid algorithm for minimizing the dual convex program. At a first glance, it may seem enough to show that  $\|\vec{\lambda}^{\star}\|\leq 2^{\frac{{\rm poly}(m)}{\eta}}$ since the number of iterations of the ellipsoid algorithm depends on $\log  \|\vec{\lambda}^{\star}\|$. Unfortunately, this is not enough since each call to the oracle with input $\vec{\lambda}$ takes time polynomial in the number of bits needed to represent $e^{-\lambda_e}$ for any $e\in [m]$. We show the following theorem which provides a polynomial bound on $\|\vec{\lambda}^{\star}\|$.

\begin{theorem}\label{thm:radius}
Let $\vec{\theta}$ be in the  $\eta$-interior of $P(\M) \subseteq  \mathbb{R}^m$. Then there exists an optimal solution $\vec{\lambda}^{\star}$ to the dual of the max-entropy convex program such that $\|\vec{\lambda}^{\star}\|\leq \frac{m}{\eta}$.
\end{theorem}

\noindent
We specifically note that the proof of this theorem needs that $\lambda^\star$ satisfies $A^=\lambda^\star=0.$ Combinatorially, it is an interesting open problem to see one can get such a bound depending  only  on $\nfrac{1}{\eta}.$
Next we generalize Theorem~\ref{thm:strong} to polytopes where only an approximate counting oracle exists, for example, the perfect matching problem in bipartite graphs. While we  state this theorem  in the context of deterministic counting oracles, it holds in the randomized setting as well.

\begin{theorem}\label{thm:approximate-count}
There exists an algorithm,  that given a maximal set of linearly independent equalities $(A^=,b),$  a generalized  approximate  counting oracle for  $P(\M) \subseteq \mathbb{R}^m$, a $\vec{\theta}$ in $\eta$-interior of $P(\M)$ and an $\epsilon>0,$  returns a $\vec{\lambda}^\circ$ such that
$$f_\theta({\vec{\lambda}}^\circ) \leq f_\theta(\vec{\lambda}^\star) + \eps.$$
Here $\vec{\lambda}^{\star}$ is an optimal solution to the dual of the max-entropy convex program for $(\M,\vec{\theta}).$ Assuming that the generalized approximate counting oracle is polynomial in its input parameters, the running time of the algorithm is polynomial in $m$, $\nfrac{1}{\eta}$, $\nfrac{1}{\epsilon}$  and the number of bits needed to represent $\theta$ and  $(A^=,b).$
\end{theorem}

\noindent
It can be shown that once we have a solution $\lambda^\circ$ to the dual convex program such that
$$f_\theta({\vec{\lambda}}^\circ) \leq f_\theta(\vec{\lambda}^\star) + \eps$$
as in Theorems \ref{thm:strong} and \ref{thm:approximate-count}, one can show that the marginals obtained from the distribution corresponding to $\lambda^\circ$ is close to that of $\lambda^\star$ (which is $\theta$), i.e.,
$\left\| \theta^{{{\vlambda}^\circ}}- \theta \right\|_\infty \leq O(\sqrt{\epsilon}).$
See Appendix \ref{sec:omit-near-optimal}, and in particular   Corollary  \ref{cor:marginal} for a proof.

\begin{remark}\label{rem:grad}
We can also obtain proofs of Theorems \ref{thm:strong} and \ref{thm:approximate-count} by applying the framework of projected gradient descent. (See Section 3.2.3 in \cite{Nesterov} for   details on the gradient descent method.)
For the gradient descent method to be polynomial time, one would need an upper bound on  $\|\lambda^\star\|$ and $\|\nabla f\|_{2 \rightarrow 2}.$ The first bound is provided by Theorem \ref{thm:radius} and the second bound is proved in Theorem \ref{thm:2to2} in Appendix \ref{sec:grad}. We have chosen the ellipsoid method-based proofs of Theorems \ref{thm:strong} and \ref{thm:approximate-count} since the ellipsoid method is required in the proof of Theorem \ref{thm:reverse}.   
\end{remark}

\noindent
Our final theorem proves the reverse:  If one can compute good approximations to  the max-entropy convex program for $P(\M)$  for a given marginal vector, then one can compute good approximations to the number of vertices in $P(\M).$
First, we need a notion of a {\em max-entropy oracle} for $\M.$

\begin{definition}\label{def:approxoracle} An {\em approximate max-entropy oracle} for $\M,$ given a  $\vtheta$ in the $\eta$-interior of $P(\M),$  a $\zeta >0$, and an $ \epsilon >0$, either
\begin{enumerate}
\item asserts that $\inf_{\vlambda}f_{\vtheta}(\vlambda)\geq \zeta -\epsilon$ or
\item  returns a $\vlambda\in \RR^{m}$ such that $f_{\vtheta}(\vlambda)\leq \zeta+\epsilon$.
\end{enumerate}
The oracle is assumed to be efficient, i.e., it runs in time polynomial in $m$, $\nfrac{1}{\epsilon},$ $\nfrac{1}{\eta}$ and the number of bits needed to represent $\zeta.$
\end{definition}

\noindent
This is consistent with the algorithms given by Theorem \ref{thm:strong} and \ref{thm:approximate-count}.

\begin{theorem}\label{thm:reverse}
There exists an algorithm that, given a maximal set of linearly independent equalities $(A^=,b)$ and a separation oracle and an  approximate optimization oracle for  $\M$ as above,  returns a $\widetilde{Z}$ such that $(1-\epsilon)|\M|\leq \widetilde{Z}\leq (1+\epsilon)|\M|.$ Assuming that the running times of the separation oracle and the approximate max-entropy oracle are polynomial in their respective  input parameters, the  running time of the algorithm is bounded by a polynomial in $m,$ $\nfrac{1}{\epsilon}$ and the number of bits needed to represent $(A^=,b).$
\end{theorem}

\noindent
Analogously, one can easily formulate and prove a randomized version of Theorem \ref{thm:reverse}, we omit the details.
As an important corollary of this theorem, if one is able to efficiently find approximate max-entropy distributions for the perfect matching polytope for general graphs, then one can approximately count the number of perfect matchings they contain. Both problems have long been open and this result, in particular, relates their hardness.

\begin{remark}\label{rem:gencount}
One may ask if Theorem \ref{thm:reverse} can be strengthened to obtain  generalized approximate counting oracles from max-entropy oracles. The question is natural since Theorems \ref{thm:strong} and \ref{thm:approximate-count} assume access to generalized counting oracles. The answer is yes and is provided in Theorem \ref{thm:reverse2} in Appendix \ref{sec:general}. It turns out that one needs access to a generalized max-entropy oracle, an oracle that can compute the distribution that minimizes the KL-divergence with respect to a fixed product distribution and a given set of marginals. These latter programs are shown, in Appendix \ref{sec:general}, to be no more general than max-entropy programs. In fact, analogs of Theorems \ref{thm:strong} and \ref{thm:approximate-count} can be proved for min-KL-divergence programs rather than max-entropy programs, see Theorem \ref{thm:approximate-count1}.  
\end{remark}

\subsection{The Ellipsoid Algorithm} \label{sec:ellipsoid}
In this section we review the basics of the ellipsoid algorithm.
The ellipsoid algorithm is used in the proofs  of our equivalence between optimization and counting: Both in the proof of Theorems \ref{thm:strong} and \ref{thm:approximate-count} and in the proof of Theorem \ref{thm:reverse}.
Consider the following optimization problem where  $g(\cdot)$ is convex and $h_i(\cdot)$ are affine functions.
\begin{eqnarray*}
 \inf &  g({\vec{\lambda}})  \\ \nonumber
 \textrm{s.t.}&\\
\forall \; {1\leq i\leq k} & h_i(\vec{\lambda})=\vec{0}  \\
\nonumber & \vec{\lambda}\in \mathbb{R}^m
\end{eqnarray*}
We assume that $g$ is differentiable everywhere and that its gradient, denoted by $\nabla g,$ is defined everywhere.
In our application, for a polytope $P(\M)$ and a $\theta$ in the $\eta$-interior of $P(\M),$  $g=f_\theta,$   the objective function in the dual program of Figure \ref{conv-dual-entropya}. The $h_i(\cdot)$s are the constraints $A^=\vec{\lambda}=\vec{0},$ where $(A^=,\vec{c})$ is the maximal set of linearly independent equalities satisfied by the vertices of $\M.$ Thus, as noted in Lemma \ref{lem:shift0},  we can restrict our search for the optimal solution to the set  $K$ which is defined to be
$$K\defeq \{\vec{\lambda}\in \mathbb{R}^m:  A^=\vec{\lambda}=\vec{0}\}.$$
Note that  $\vec{0} \in K.$
The ellipsoid algorithm can be used to solve such a convex program under fairly general conditions
and we first state a version of it needed in the proof of Theorem \ref{thm:strong}. A crucial requirement is a
{\em  strong first-order oracle} for $g$ which is a function such that given a $\vec{\lambda},$ outputs $g(\vec{\lambda})$ and $\nabla g (\vec{\lambda}).$ Since we are only interested in  $\vec{\lambda} \in K,$  and we are given the equalities describing $K$ explicitly, we assume that we can project  $\nabla g (\vec{\lambda})$ to $K.$ By abuse of notation, we denote the latter also by $\nabla g(\vec{\lambda}).$

The following theorem claims that if one is given access to a strong first-order oracle for $g$, one can use the ellipsoid algorithm to obtain an approximately optimal solution to the convex program mentioned above. This statement is easily derivable from \cite{BentalN12} (Theorem 8.2.1).

\begin{theorem}\label{thm:strong-ellipsoid}
Given any $\beta>0$ and  $R>0,$ there is an algorithm which, given a strong first-order oracle for $g,$    returns a point ${\vec{\lambda}}'\in \mathbb{R}^m$ such that $$g({\vec{\lambda}}')
\leq  \inf_{\vec{\lambda} \in K, \; \|\vec{\lambda}\|_\infty \leq R} g(\vec{\lambda})+\beta \left(\sup_{\vec{\lambda} \in K, \; \|\vec{\lambda}\|_\infty \leq R} g(\vec{\lambda})-\inf_{\vec{\lambda} \in K, \; \|\vec{\lambda}\|_\infty \leq R} g(\vec{\lambda}) \right).$$
The number of calls to the strong first-order oracle for $g$  are bounded by a polynomial in $m$, $\log R$ and  $\log{\nfrac{1}{\beta}}.$
\end{theorem}

\noindent
While we do not explicitly describe the ellipsoid algorithm here, we need the following basic properties about minimum volume enclosing ellipsoids which  forms the basis of the ellipsoid algorithm. A set $E\subseteq \RR^{m}$ is an ellipsoid if there exists a vector $\vec{a}\in \RR^{m}$ and a positive definite $m \times m$-matrix $A$ such that
$E=E(A,a)\defeq \{\vec{x}\in \RR^{m}: (\vec{x}-\vec{a})^{\top}A^{-1}(\vec{x}-\vec{a})\leq 1\}.$ We also denote $\textrm{Vol}(E)$ to be the volume enclosed by the ellipsoid $E$. The following theorem follows from the L\"owner-John Ellipsoid. We refer the reader  to \cite{GrotschelLS88} for more details.

\begin{theorem}\label{thm:john}
Given an ellipsoid $E(A,\vec{a})$ and a half-space $\{\vec{x}:\iprod{\vec{c}}{\vec{x}}\leq \iprod{\vec{c}}{\vec{a}}\}$ passing through $\vec{a}$ there exists an ellipsoid $E'\supseteq E(A,\vec{a})\cap \{\vec{x}:\iprod{\vec{c}}{\vec{x}}\leq \iprod{\vec{c}}{\vec{a}}\}$ such that $\frac{\mathrm{Vol}(E')}{\mathrm{Vol}(E)}\leq e^{-\frac{1}{2m}}$.
\end{theorem}

\noindent
In our applications of this theorem, we  in fact need the ellipsoid to be in an affine space of dimension possibly lower than $m.$ The definitions and the theorem continue to hold under such a setting.

\section{Examples of Combinatorial Polytopes}\label{sec:examples}
If one wishes, one can keep  the following  combinatorial polytopes in mind while trying to understand and interpret the results of this paper.

\vspace{2mm}
\noindent
{\bf The Spanning Tree Polytope.} Given a graph $G=(V,E)$, let $$\M\defeq \left\{\one_T\in \RR^{|E|}: T\subseteq E \textrm{ is a spanning tree of }G \right\}.$$ It follows from a result of Edmonds~\cite{Edmonds71} that $$P(\M)=\left\{\vec{x}\in \RR^{|E|}_{\geq 0}: x(E(V))=|V|-1, \ x(E(S))\leq |S|-1\; \; \forall S\subseteq V\right\}$$ where, for $S \subseteq V,$ $E(S)\defeq \{e=\{u,v\}\in E:\{u,v\}\cap S=\{u,v\}\}$ and, for a subset of edges $H \subseteq E,$  $x(H) \defeq \sum_{e \in H} x_e.$
Edmond \cite{edmonds1970submodular} also shows the existence of a separation oracle for this polytope. A generalized exact counting oracle is known for this spanning tree polytope  via  Kirchoff's matrix-tree theorem, see \cite{GodsilR01}.

\vspace{2mm}
\noindent
{\bf The Perfect Matching Polytope for Bipartite Graphs.} Given a bipartite graph $G=(V,E)$, let $$\M \defeq \left\{\one_M\in \RR^{|E|}:M \textrm{ is a perfect matching in } G \right\}.$$ It follows from a theorem of Birkhoff~\cite{Birkhoff46} that, when $G$ is bipartite, $$P(\M)=\left\{\vec{x}\in \RR^{|E|}_{\geq 0}: x(\delta(v))=1\; \; \forall v\in V\right\}$$
where, for $v \in V,$ $\delta(v) \defeq   \{e=\{u,v\}\in E\}.$
Here, it can be shown that all the facets, i.e., the defining inequalities, are one of the set of $2m$ inequalities $0 \leq x_e \leq 1$ for all $e \in [m].$
The exact counting problem is \#P-hard and while a (randomized)  generalized approximate counting oracle follows from a result of Jerrum, Sinclair and Vigoda ~\cite{JerrumSV04} for computing permanents.

\vspace{2mm}
\noindent
{\bf The Cycle Cover Polytope for Directed Graphs.}
Given a directed graph $G=(V,A),$ let
$$\M \defeq \left\{\one_M\in \RR^{|A|}:M \textrm{ is a cycle cover in } G \right\}.$$
A cycle cover in $G$ is a collection of vertex disjoint directed cycles that cover all the vertices of $G.$
The corresponding cycle cover polytope is denoted by $P(\M).$
This polytope is easily seen to be a  special case of the perfect matching polytope for bipartite graphs as follows. For $G=(V,A),$ construct a bipartite graph $H=(V_L,V_R,E)$ where $V_L=V_R=V.$ For each vertex $v \in V$ we have $v_L \in V_L$ and $v_R \in V_R.$ There is an edge between $u_L \in V_L$ and  $v_R \in V_R$ in $H$ if and only if $(u,v) \in A.$ Thus, there is a one-to-one correspondence between cycle covers in $G$ and perfect matchings in $H.$ Hence, the  \cite{JerrumSV04} algorithm gives a generalized approximate counting oracle  in this case as well.

\vspace{2mm}
\noindent
{\bf The Perfect Matching Polytope for General Graphs.} Given a graph $G=(V,E)$, let $$\M\defeq \left\{\one_M\in \RR^{|E|}:M \textrm{ is  a perfect matching in } G\right\}.$$ A celebrated result of Edmonds~\cite{Edmonds65b} states that
\begin{eqnarray*}
P(\M)&=&\left\{\vec{x}\in \RR^{|E|}_{\geq 0}: x(\delta(v)) =  1\; \forall v\in V,  \;
 x(E(S))\leq \frac{|S|-1}{2}\; \forall S \subseteq V, |S| \ \  \textrm {odd }\right\}.
\end{eqnarray*}
The separation oracle for this polytope is non-trivial and follows from the characterization result of Edmonds. A direct separation oracle was also given by Padberg and Rao \cite{padberg1982odd}.
 Coming up with a counting oracle for this polytope, even with uniform weights which counts the number of perfect matchings in a general graph, is a long-standing open problem.

\section{New Algorithmic Approaches for the Traveling Salesman Problem}\label{sec:tsp}
Max-entropy distributions over spanning trees have been successfully applied to obtain improved algorithms for the symmetric \cite{GharanSS11} as well as the asymmetric traveling salesman problem \cite{AsadpourGMGS10}. We  outline here a different algorithmic approach, using max-entropy distributions over cycle covers, which becomes computationally feasible as a consequence of our results.
Let us consider the asymmetric traveling salesman problem (ATSP). We are given a complete directed graph $G=(V,E)$ and cost function $c:E\rightarrow \RR_{\geq 0}$ which satisfies the directed triangle inequality. The goal is find a Hamiltonian cycle of smallest cost. First, we formulate the following subtour elimination linear program in Figure \ref{fig:HK}.
 
\begin{figure}[h]
\begin{center}
\begin{eqnarray*}
\min & \sum_{e\in E}c_ex_e  \nonumber\\
\st \nonumber \\
\forall v\in V & x(\delta^{+}(v))=x(\delta^{-}(v))=1  \label{consvertex}\\
\forall S\subseteq V & x(\delta^{+}(S))\geq 1 \label{consset}\\
\forall e\in E & 0\leq x_e \leq 1 \label{consset2}
\end{eqnarray*}
\end{center}
\caption{Subtour Elimination LP for $G=(V,E)$ and $c$} \label{subtour}\label{fig:HK}
\end{figure}
 
\noindent
Here, for a vertex $v,$ $\delta^+(v)$ is the set of directed edges going out of it and $\delta^-(v)$ is the set of directed edges coming in to $v.$  Let $x^{\star}$ denote the optimal solution to this linear program. The authors of \cite{AsadpourGMGS10} make the observation that $\theta_{uv} \defeq \frac{n-1}{n}(x^{\star}_{uv}+x^{\star}_{vu})$ defined on the undirected edges is a point in the interior of the spanning tree polytope on $G$. The algorithm then samples a spanning tree $T$ from the max-entropy distribution with marginals as given by $\theta$ and crucially relies on properties of such a $T$ to obtain an $O\left(\frac{\log n}{\log \log n}\right)$-approximation algorithm for the ATSP problem.
 
Interestingly, there is another integral polytope in which $x^{\star}$ is contained. Consider the convex hull $P$ of all cycle covers of $G$, see Section \ref{sec:examples}. Then,
$$P=\left\{x\in \RR^{|E|}_{\geq 0}: x(\delta^{+}(v))=x(\delta^-(v))=1\right\}.$$
It is easy to see that $x^{\star} \in P$. Similar to the cycle cover algorithm of Frieze et al~\cite{frieze1982worst}, the following is a natural algorithm for the ATSP problem.
 
\vspace{2mm}
\noindent
\textbf{Randomized Cycle Cover Algorithm}
\begin{enumerate}
\item Initialize $H\gets \emptyset$.
\item {\bf While} $G$ is not a single vertex
\begin{itemize}
\item Solve the subtour elimination LP for $G$  to obtain the solution $x^{\star}$.
\item Sample a cycle cover $\C$ from the max-entropy distribution with marginals $x^{\star}$.
\item Include in $H$ all edges in $\C$, i.e., $H\gets H\cup \left( \cup_{C\in \C} C \right).$
\item Select one representative vertex $v_C$ in each cycle $C\in \C$ and delete all the other vertices.
\end{itemize}
\item {\bf Return} $H$.
\end{enumerate}
 
\noindent
 Before analyzing the performance of this algorithm, a basic question is whether this algorithm can be implemented in polynomial time. As an application of Theorem~\ref{thm:approximate-count} to the cycle cover polytope for directed graphs, it follows that one can sample a cycle cover from the max-entropy distribution in polynomial time and, thus, the question is answered affirmatively. The generalized (randomized) approximate counting oracle for cycle covers in a graph follows from the work of \cite{JerrumSV04}. The technical condition of interiority of $x^{\star}$ can be satisfied with a slight loss in optimality of the objective function. The analysis of worst case performance of this algorithm is left open, but to the best of our knowledge, there is no example ruling out that the Randomized Cycle Cover Algorithm is an $O(1)$-approximation. Similarly, the application of Theorem \ref{thm:approximate-count} to the perfect matching polytope in bipartite graphs makes the permanent-based approach suggested in \cite{Vishnoi12} for the (symmetric) TSP computationally feasible.

\section{Bounding Box}\label{sec:bounding}
In this section, we prove Theorem \ref{thm:radius} and show that there is a bounding box of small radius containing the optimal solution $\vlambda^{\star}$. We begin with the following lemma.

\begin{lemma}\label{lem:theta-lambda}
Let $\vtheta$ be a point in the $\eta$-interior of $P(\M) \subseteq \RR^m$ and let $\vlambda^\star$ be the optimal solution to the dual convex program.
Then for any $\vec{x}\in P(\M)$
$$ \langle \vlambda^\star,\vtheta - \vec{x}\rangle \leq m.$$
\end{lemma}

\begin{proof}
First, note that the supremum of the primal convex program over all $\vtheta$ is $\ln |\M| \leq m.$
Hence, from  strong duality it follows that
$f(\vlambda^\star) \leq m.$ This implies that
$$ f(\vlambda^\star) = \langle \vlambda^\star, \vtheta \rangle + \ln \sum_{M \in \M} e^{-\lambda^\star(M)} = \ln \sum_{M \in \M} e^{ \langle \vlambda^\star, \vtheta \rangle-\lambda^\star(M)} \leq m.$$
Hence, for every $M \in \M,$
\begin{equation}\label{eq4-1}
 \langle \vlambda^\star, \vtheta \rangle-\lambda^\star(M) \leq m.
 \end{equation}
Since $\vec{x}\in P(\M)$, we have $\vec{x}=\sum_{M\in \M} r_{M} \one_M$ where $\sum_{M\in \M}r_M=1$ and $r_{M}\geq 0$ for each $M\in \M$. Multiplying \eqref{eq4-1} by $r_M$ and summing over $M$ we get
$$\sum_{M\in \M}r_M \left( \langle \vlambda^\star, \vtheta \rangle-\lambda^\star(M) \right) \leq \sum_{M\in \M}r_M  m.$$
This implies that  $$\langle \vlambda^\star, \vtheta \rangle- \sum_{M\in \M}r_M\lambda^\star(M) \leq m.$$
As a consequence, we obtain that $\langle \vlambda^\star, \vtheta \rangle- \langle  \vlambda^{\star} ,\vec{x}\rangle \leq m,$ completing the proof of the lemma.
\end{proof}

\paragraph{\bf Proof of Theorem~\ref{thm:radius}.}
Recall that $A^={\vec{x}}=\vec{c}$ denotes the maximal set of independent equalities satisfied by $P(\M)$. We now define the following objects.
Let $$B\defeq \{\vec{x}\in \RR^m:A^={\vec{x}}=\vec{c},\; \|\vec{x}-\vtheta\|\leq \eta\}$$ be the ball centered around $\vtheta$ restricted to the affine space $A^=\vec{x}=\vec{c}$ of radius $\eta$. Since $\vtheta$ is in the $\eta$-interior of $P(\M),$ we have that $B\subseteq P(\M)$. Let $$Q\defeq \left\{\vec{y}\in \RR^m : A^=\vec{y}=\vec{c},\; \; \|\vec{y}-\vtheta\|\leq \nfrac{1}{\eta}\right\}$$ be the ball centered around $\vtheta$ of radius $\frac{1}{\eta}$ in the same affine space and let $$\widetilde{Q}\defeq \{\vec{z}\in \RR^m : A^=\vec{z}=\vec{c}, \; \; \iprod{\vec{z}-\vtheta}{\vec{x}-\vtheta}\leq 1\;\; \forall \vec{x}\in B\}.$$

\begin{lemma}
$Q=\widetilde{Q}$.
\end{lemma}
\begin{proof}
We first prove that {\bf $Q\subseteq \widetilde{Q}.$} Let $\vec{y}\in Q$. The constraints $A^=\vec{y}=\vec{c}$ are clearly satisfied since $\vec{y}\in Q$. For any
$\vec{x}\in B$, $$\iprod{\vec{y}-\vtheta}{\vec{x}-\vtheta}\leq \|\vec{y}-\vtheta\|\|\vec{x}-\vtheta\|\leq \frac{1}{\eta}\cdot \eta= 1.$$ Thus, $\vec{y}\in \widetilde{Q}$.
Now we show that {\bf $Q\subseteq \widetilde{Q}.$} Let $\vec{z}\in \widetilde{Q}$. The constraints $A^=\vec{z}=\vec{c}$ are clearly satisfied since $\vec{z}\in \widetilde{Q}$. Now consider
$$\vec{z}' \defeq \vtheta +\frac{\vec{z}-\vtheta}{\|\vec{z}-\vtheta\|}\cdot \eta.$$ We have that $$A^=\vec{z}'= A^=\vtheta +\frac{A^=\vec{z}-A^=\vtheta}{\|\vec{z}-\vtheta\|}\cdot \eta = \vec{c}.$$ Moreover, $$\|\vec{z}'-\vtheta\|=\left\|\vtheta +\frac{\vec{z}-\vtheta}{\|\vec{z}-\vtheta\|}\cdot \eta-\vtheta \right\|=\frac{\|\vec{z}-\vtheta\|}{\|\vec{z}-\vtheta\|}\cdot \eta=\eta.$$ Thus, $\vec{z}'\in B$. Hence, we must have
$
\iprod{\vec{z}-\vtheta}{\vec{{z}'}-\vtheta} \leq 1.$
This implies that
$$\left\langle \vec{z}-\vtheta,\frac{\vec{z}-\vtheta}{\|\vec{z}-\vtheta\|}\cdot \eta \right\rangle \leq 1$$
and, therefore,
$$\|\vec{z}-\vtheta\|\leq \nfrac{1}{\eta}.$$
Thus, $\vec{z}\in Q$ completing the proof.
\end{proof}

\noindent
We now show that $$\vec{\widetilde{\lambda}}=\nfrac{-\vec{\lambda}^{\star}}{m}+\vtheta\in \widetilde{Q}.$$ To see this, first observe that $$A^=\vec{\widetilde{\lambda}}=\nfrac{-A^={\vec{\lambda}}^{\star}}{m}+A^=\vtheta= \vec{0}+\vec{c}.$$ Here we have used the fact that $A^=\lambda^\star = 0,$ see Lemma \ref{lem:shift0}. We now verify the second condition. Let $\vec{x}\in B$. Then
$$\iprod{\vec{\widetilde{\lambda}}-\vtheta}{\vec{x}-\vtheta}=\frac{- \iprod{\vec{\lambda}^{\star}}{\vec{x}-\vtheta}}{m}\leq \frac{1}{m} \cdot m=1$$ where the last inequality follows from that fact that $\vec{x}\in B\subseteq P(\M)$ and Lemma~\ref{lem:theta-lambda}.
Thus, $\vec{\widetilde{\lambda}}\in Q$ and therefore, we must have $\|\vec{\widetilde{\lambda}}-\vtheta\|\leq \nfrac{1}{\eta}$. Therefore, $\|\nfrac{\vec{\lambda}^{\star}}{m}\|\leq \nfrac{1}{\eta}$ proving Theorem~\ref{thm:radius}.

\section{Optimization via Counting}\label{sec:algorithm}
In this section, we prove Theorems~\ref{thm:strong} and \ref{thm:approximate-count}. The proof of both the theorems rely on the bounding box result of Theorem \ref{thm:radius} and employs the framework of the ellipsoid algorithm from Section \ref{sec:ellipsoid}.

\subsection{Proof of  Theorem \ref{thm:strong}}
We first use Theorem~\ref{thm:strong-ellipsoid} to give a proof of Theorem~\ref{thm:strong}. The algorithm assumes access to a strong first-order oracle for $P(\M).$ We then present details of how to implement a strong first-order oracle using an  generalized exact counting oracle.
Suppose $\vlambda^\star$ is the optimum of our convex program. Theorem~\ref{thm:radius} implies that for  $\| {\vlambda}^{\star}\|_\infty \leq \nfrac{m}{\eta}.$ Thus, we may pick the bounding radius to $R \defeq \nfrac{m}{\eta}$ and it does not cut the optimal $\lambda^\star$ we are looking for. The only thing left to choose is a $\beta$ such that
$$ \beta\leq \frac{\eps}{ \left(\sup_{\vlambda \in K, \; \|\vlambda\|_\infty \leq R} f_\theta(\vlambda)-\inf_{\vlambda \in K, \; \|\vlambda\|_\infty \leq R} f_\theta(\vlambda) \right) }.$$
This would imply that the solution ${\vlambda}^\circ$ output by employing the ellipsoid method from Theorem \ref{thm:strong-ellipsoid} is such that
$$ f_\theta({\vlambda}^\circ)  \leq f_\theta(\vlambda^\star) + \eps.$$

\noindent
To establish a bound on $\beta,$ start by noticing that $\inf_{\vlambda} f_\theta(\vlambda) \geq 0.$ This follows from weak-duality and the fact that entropy is always non-negative.
On the other hand we have the following simple lemma.
\begin{lemma}
$\sup_{ \|\vlambda\|_\infty \leq R} f_\theta(\vlambda)\leq (2m+1)R.$
\end{lemma}
\begin{proof}
$$f_\theta(\vlambda) \leq | \langle \vlambda, \vtheta \rangle | + \left| \ln \sum_{M \in \M} e^{-\lambda(M)} \right| \leq mR + \ln (2^m e^{mR}) \leq (2m+1)R.$$
Here we have used the fact that $|\M| \leq 2^m$ and that $\theta_e \in [0,1]$ for each $e \in [m].$
\end{proof}
\noindent
Thus, $\beta$ can be chosen to be $\frac{\eps}{(2m+1)R}.$
Hence, the running time of the ellipsoid method  depends polynomially on the the time it takes to implement the strong first-order oracle for $f_\theta$ and $\log \frac{mR}{\eps}.$

\noindent
Since $f_\theta(\vlambda) = \langle \vlambda, \vtheta \rangle + \ln \sum_{M \in \M} e^{-\lambda(M)}=\langle \vlambda, \vtheta \rangle + \ln Z^\vlambda,$
it is easily seen that
$$\nabla f_\theta(\vlambda)_e = \theta_e - \frac{\sum_{M \in \M, \; M \ni e} e^{-\lambda(M)}}{\sum_{N \in \M} e^{-\lambda(N)}}=\theta_e -\frac{Z^\vlambda_e}{Z^\vlambda} = \theta_e - \theta_e^\vlambda.$$
Hence, $\nabla f_\theta(\vlambda) = \vtheta - \vtheta^\vlambda.$
Recall that the strong first-order oracle for $f_\theta$ requires, for a given $\vlambda,$ $f_\theta(\vlambda)$ and $\nabla f_\theta(\vlambda).$ The generalized exact counting oracle for $P(\M)$ immediately does it as it gives us $Z^\vlambda_e$ for all $e \in [m]$ and $Z^\vlambda.$ This allows us to compute $f_\theta(\vlambda)$ and $\nabla f_\theta(\vlambda)$ in one call to such an oracle. In addition we also need time proportional to the number of bits needed to represent $\theta.$

Thus, the number of calls to the counting oracle by the ellipsoid algorithm of Theorem \ref{thm:strong-ellipsoid} is bounded by a polynomial in $m,$ $\log R$ and $\log \nfrac{1}{\eps}.$ Since each oracle call can be implemented in time polynomial in $m$ and $R$,\footnote{Here we ignore the fact that $e^{-\lambda_e}$ can be irrational. This issue can be dealt in a standard manner as is done in the implementation details of all ellipsoid algorithms. See \cite{GrotschelLS88} for details.} this gives the required running time and concludes the proof of Theorem \ref{thm:strong}.

\def \R   {{\mathbb R}}
\def \Q   {{\mathbb Q}}

\subsection{Proof of  Theorem \ref{thm:approximate-count}}

Now we give the ellipsoid algorithm that works with a generalized approximate counting oracle and prove Theorem~\ref{thm:approximate-count}.
Here, the fact that the counting oracle is approximate means that the gradient computed as in the previous section is approximate. Thus, this raises the possibility of cutting off the optimal $\lambda^\star$ during the run of the ellipsoid algorithm.
We present the ellipsoid algorithm to check, given a $\theta$ and a $\zeta,$ whether $|f_\theta(\lambda^\star) - \zeta| \leq \eps$.  The technical heart of the matter is to show that when $| f_\theta(\lambda^\star)- \zeta| \leq \eps,$ $\lambda^\star$ is  never cut off of the successive ellipsoids obtained by adding the approximate gradient constraints. Moreover, in this case,  once the radius of the ellipsoid becomes small enough, we can output its center as a guess for $\lambda^\star.$
Since the radius of the final ellipsoid is small and contains $\lambda^\star$, the following lemma, which bounds the Lipschitz constant of  $f_\theta,$ implies that the value of $f_\theta$ at the center of the ellipsoid is close enough to $f_\theta(\lambda^\star).$
\begin{lemma}\label{lem:dual-nearoptimal}
For any $\vec{\lambda},\vec{\lambda}'$
$$f_\theta(\vec{\lambda})-f_\theta(\vec{\lambda}')\leq 2\sqrt{m} \|\vec{\lambda}-\vec{\lambda}'\|.$$
\end{lemma}
\begin{proof}
We have
\begin{align*}
f_\theta(\vec{\lambda})-f_\theta(\vec{\lambda}')&= \iprod{\vec{\theta}}{\vec{\lambda}-\vec{\lambda}'}+ \ln \frac{\sum_{M\in\M}e^{-\lambda(M)}}{\sum_{M\in\M}e^{-\lambda'(M)}}\\
& \leq  \|\vec{\theta}\|\|\vec{\lambda}-\vec{\lambda}'\|+ \ln \max_{M\in \M} \frac{e^{-\lambda(M)}}{e^{-\lambda'(M)}}\\
&\leq \sqrt{m} \|\vec{\lambda}-\vec{\lambda}'\|+ \max_{M\in \M} (\lambda'(M)-\lambda(M))\\
& \leq  \sqrt{m}\|\vec{\lambda}-\vec{\lambda}'\|+  \sqrt{m} \|\vec{\lambda}-\vec{\lambda}'\|\leq 2\sqrt{m} \|\vec{\lambda}-\vec{\lambda}'\|
\end{align*}
which completes the proof. Here we have used the Cauchy-Schwarz inequality in the first and third inequalities and in the second inequality we have used the fact that $\theta \in [0,1]^m.$
\end{proof}

\noindent
Proceeding to the ellipsoid algorithm underlying the proof of Theorem \ref{thm:approximate-count}, we do a binary search on the optimal value $f_\theta(\vlambda^{\star})$ up to an accuracy of $\nfrac{\epsilon}{8}$. For a guess $\zeta\in(0,m]$, we check whether the guess is correct with the following ellipsoid algorithm.\footnote{For ease of analysis, we assume that all oracle calls are answered correctly with probability $1.$ The failure probability can be adjusted to arbitrary precision with a slight degradation in the running time.}
\begin{enumerate}
\item {\bf Input}
\begin{enumerate}

\item An error parameter $\eps>0.$
\item An interiority parameter $\eta>0.$
\item A $\theta$ which is guaranteed to be in the $\eta$-interior of $P(\M).$
 \item A maximally linearly independent set of equalities $(A^=,b)$ for $P(\M).$
\item A generalized approximate counting oracle for $\M.$
\item A guess $\zeta \in (0,m]$ (for  $f_\theta(\lambda^\star)).$
\end{enumerate}

\item {\bf Initialization}
\begin{enumerate}
\item  Let  $E_0 \defeq E(B_{0},\vec{c}_{0})$ be a sphere with radius $R=\nfrac{m}{\eta}$ centered around the origin (thus, containing $\vlambda^{\star}$ by Theorem \ref{thm:radius}) and  restricted to  the affine space $A^=x=0.$
\item Set $t=0$
\end{enumerate}

\item {\bf Repeat} until  the ellipsoid $E_t$ is contained in a ball of radius at most $\frac{\eps}{16 \sqrt{m}}.$
\begin{enumerate}
\item Given the ellipsoid  $E_{t} \defeq E(B_{t},\vec{c}_{t})$, set $\vlambda_t \defeq \vec{c}_{t}.$

\item\label{forward-1} Compute $\zeta_t$ using the counting oracle such that $f_\theta(\vlambda_t)-\nfrac{\epsilon}{8}\leq \zeta_t\leq f_\theta(\vlambda_t)+\nfrac{\epsilon}{8}.$
\item  {\bf If}  $\zeta_t \leq \zeta$
\begin{enumerate}
 \item  {\bf then}  {\bf return} $\vlambda_t$ and {\bf stop}.
\item  {\bf else}
\begin{enumerate}
\item\label{forward-2} Compute $\theta_t$ such that  $\|\vtheta_t-\vtheta^{\vlambda_t}\|_1\leq \nfrac{\epsilon}{16R}$ using the counting oracle.
\item\label{forward-4}  Compute the ellipsoid $E_{t+1}$ to be the smallest ellipsoid containing the half-ellipsoid $\{\vlambda\in E_{t}: \langle \vlambda-\vlambda_t,  \vtheta-{\vtheta_t} \rangle\leq 0\}$  restricted to  the affine space $A^=x=0.$
\end{enumerate}
\item $t=t+1.$
\end{enumerate}
\end{enumerate}
\item\label{forward-3} Let $T=t$ and compute $\zeta_T$ using the counting oracle such that $f_\theta(\vlambda_T)-\nfrac{\epsilon}{8}\leq \zeta_T\leq f_\theta(\vlambda_T)+\nfrac{\epsilon}{8}$.
\item {\bf If}  $\zeta_T \leq \zeta$
\begin{enumerate}
 \item  {\bf then}  {\bf return} $\vlambda_T$ and {\bf stop}.
\item  {\bf else} {\bf return} $f_\theta(\lambda^\star) > \zeta$ and {\bf stop}  ($\zeta$ is not a good guess for $f_\theta(\lambda^\star)$).
\end{enumerate}
\end{enumerate}

\noindent
We first show that the algorithm can be implemented using a polynomial number of queries to the approximate oracle. Steps (\ref{forward-1}), (\ref{forward-2}) and (\ref{forward-3}) can be computed using oracle calls to the generalized approximate counting oracle for $\M$ to obtain  $\widetilde{Z}^{\vlambda_t}$ such that $(1-\nfrac{\epsilon}{16})Z^{\vlambda_t}\leq \widetilde{Z}^{\vlambda_t}\leq (1+\nfrac{\epsilon}{16})Z^{\vlambda_t}.$ We set $\zeta_t=\langle \vlambda_t,\vtheta\rangle + \ln{ \widetilde{Z}^{\vlambda_t}}$. A simple calculation then  shows that
 $$f_\theta(\vlambda_t)-\frac{\epsilon}{8}\leq \zeta_t\leq f_\theta(\vlambda_t)+\frac{\epsilon}{8}$$
 since $f_\theta(\vlambda_t)=\langle \vlambda_t,\vtheta \rangle+ \ln{ {Z}^{\vlambda_t}}.$ Similarly using one oracle call to the counting oracle with error parameter $\frac{\epsilon}{16R}$, we can compute  $$\|\vtheta_t-\vtheta^{\vlambda_t}\|_1\leq \frac{\epsilon}{8R}$$ as needed in Step (\ref{forward-2}) of the algorithm.
Using Theorem~\ref{thm:john}, the number of iterations can be bounded by a polynomial in $m, \log \nfrac{R}{\epsilon}$. The analysis is quite standard and omitted. Each of the oracle call can be implemented in time polynomial in $m, R$ and $\nfrac{1}{\epsilon}$. We now show the following lemma which completes the proof of Theorem~\ref{thm:approximate-count}.

\begin{lemma}\label{lem:approximate-solution}
Let $\zeta^\circ$ be the smallest guess for which the ellipsoid algorithm succeeds in finding a solution and let $\vlambda^\circ$ denote the corresponding solution returned. Then, $f_\theta(\vlambda^\circ)\leq f_\theta(\vlambda^\star)+\epsilon$.
\end{lemma}
\begin{proof}
 Observe that we have $$f_\theta(\vlambda^\circ)-\frac{\epsilon}{8}\leq \zeta^\circ\leq f_\theta(\vlambda^\circ)+\frac{\epsilon}{8}.$$
Since $\zeta^\circ$ is the smallest guess for which the algorithm succeeds in returning an answer, it fails for some $\zeta\in [\zeta^\circ-\nfrac{\epsilon}{8},\zeta^\circ]$.  We  show that $f_\theta(\vlambda^{\star})\geq \zeta-\nfrac{\epsilon}{4}$. This suffices to prove the theorem since
$$f_\theta(\vlambda^\circ)\leq \zeta^\circ+\frac{\epsilon}{8}\leq \zeta+\frac{\epsilon}{4}\leq f_\theta(\vlambda^{\star})+\frac{\epsilon}{2}.$$

\noindent
Suppose for the sake of contradiction that
\begin{equation}\label{eq:lambda-gamma}
f_\theta(\vlambda^{\star})<\zeta-\frac{\epsilon}{4}.
\end{equation}
We then show that $\vlambda^{\star}$ must be in the final ellipsoid $E_T$ when the ellipsoid algorithm is run with guess $\zeta$. Let $\vlambda_t$ be center of the ellipsoid $E_t$ in any iteration with guess $\zeta$. Let $\zeta_t$ be computed in Step (\ref{forward-1}) such that
$$f_\theta(\vlambda_t)-\frac{\epsilon}{8}\leq \zeta_t\leq f_\theta(\vlambda_t)+\frac{\epsilon}{8}.$$
Since the algorithm does not return any answer with guess $\zeta$, we must have $\zeta_t >\zeta$. Thus,
$$f_\theta(\vlambda_t)\geq \zeta_t-\frac{\epsilon}{8}\geq \zeta-\frac{\epsilon}{8}\geq f_\theta(\vlambda^{\star})+\frac{\epsilon}{8}$$
where the last inequality follows from inequality~\eqref{eq:lambda-gamma}. Let $\vtheta_t$ be computed in Step (\ref{forward-2}) of the algorithm such that
$\|\vtheta_t-\vtheta^{\vlambda_t}\|_1\leq \frac{\epsilon}{32R}$. But then
\begin{eqnarray*}
\langle \vlambda^{\star}-\vlambda_t, \vtheta-\vtheta_t \rangle&=&\langle \vlambda^{\star}-\vlambda_t, \vtheta-\vtheta^{\vlambda_t} \rangle+\langle \vlambda^{\star}-\vlambda_t, \vtheta^{\vlambda_t}-\vtheta_{t} \rangle \\
&\leq &f_\theta(\vlambda^{\star})-f_\theta(\vlambda_{t}) + \langle \vlambda^{\star}-\vlambda_t,\vtheta^{\vlambda_t}-\vtheta_{t} \rangle \\
&\leq& -\frac{\epsilon}{8} + \|\vlambda^{\star}-\vlambda_t\|\|\vtheta^{\vlambda_t}-\vtheta_{t}\| \\
&\leq& -\frac{\epsilon}{8}+ 2R\cdot \frac{\epsilon}{32R}\leq 0
\end{eqnarray*}

\noindent
where the first inequality follows from convexity of $f$. Thus, $\vlambda^{\star}$ satisfies the separating constraint put in Step (\ref{forward-4}) of the algorithm. Therefore, it must be contained in the final ellipsoid $E_T$. Let $\vlambda_T$ be the center of the ellipsoid $E_T$. Let $\zeta_T$ be computed such that
 $| \zeta_T- f_\theta(\vlambda_T)| \leq \nfrac{\epsilon}{8}.$
Then
\begin{eqnarray*}
\zeta_T&\leq&  f_\theta(\vlambda_T)+\frac{\epsilon}{8} \leq f_\theta(\vlambda^{\star})+\sqrt{m}\|\vlambda^{\star}-\vlambda_T\|+\frac{\epsilon}{8}\\
&\leq& f_\theta(\vlambda^{\star})+\sqrt{m}\frac{\epsilon}{16\sqrt{m}} +\frac{\epsilon}{8}\leq \zeta-\frac{\epsilon}{4}+\frac{\epsilon}{4}\leq \zeta
\end{eqnarray*}

\noindent
where we use Lemma~\ref{lem:dual-nearoptimal}. Therefore, the algorithm must have returned $\lambda^\circ=\vlambda_T$ as the feasible solution for guess $\zeta$ a contradiction. This completes the proof of Lemma~\ref{lem:approximate-solution}.
\end{proof}

\def \B   {{\mathcal B}}
\section{ Counting via Optimization}\label{sec:reverse}
In this section we present the proof of Theorem \ref{thm:reverse}. We start by phrasing the problem of estimating $|\M|$ as a convex optimization problem.

\subsection{A Convex Program for Counting}
Let   $g(\vtheta)$ denote the optimum of the max-entropy program of Figure \ref{conv-entropy} for  $\M$ and a point $\theta$ in $P(\M).$ If $\theta$ is in the interior of $P(\M),$ then strong duality holds for this convex program and
$$ g(\theta) =  \inf_{\vlambda}f_{\vtheta}(\vlambda),$$
see Lemma \ref{prop:maxentropya}.
By the concavity of the Shannon entropy,  $g(\cdot)$ can be easily seen to be a concave function of $\theta.$
Recall that in the setting of Theorem \ref{thm:reverse},  we have  an access to an approximate max-entropy oracle which, given a $\theta$ in the $\eta$-interior of  $P(\M),$ a $\zeta>0$ as a guess for $g(\theta),$ and an $\eps >0,$ either  asserts that $g(\theta) \geq \zeta -\eps$ or returns a $\lambda$ such that $f_\theta(\lambda) \leq \zeta+\eps.$
  The running time of this oracle is polynomial in $m,\nfrac{1}{\eps}$ and in the number of bits needed to represent $\zeta.$
Using this oracle, we hope to get an estimate on $|\M|.$
The starting point of the proof  is the observation that the point that maximizes $g(\theta)$ when $\theta$ is in $P(\M)$ is
 $$\theta^\star \defeq \frac{1}{|\M|}\sum_{M\in \M} \one_{M},$$ the vertex centroid of $P(\M).$
\begin{lemma}\label{cl:max-entropy}
 $\sup_{\vtheta\in P(\M)} g(\theta) \leq \ln |\M|$ and  $g(\theta^\star)=\ln |\M|.$
\end{lemma}

\begin{proof}
For any $\vtheta$ in the interior of $P(\M)$, $g(\theta)$ is  the entropy of some probability distribution over the  elements in $|\M|$. A standard fact in information theory implies that the maximum entropy of any distribution over a finite set is obtained by the uniform distribution. The entropy of the uniform distribution on $\M$ is $\ln |\M|,$ hence, $g(\theta)$  can be upper bounded by $\ln { |\M|}$.  On the other hand, the uniform distribution over $\M$ has marginals equal to $\theta^\star$ and, thus, $g(\theta^\star) = \ln |\M|.$
\end{proof}

\noindent
Thus, if we could find $g(\theta^\star),$ then we can estimate $|\M|.$
Finding $g(\theta^\star)$ is the same as solving the convex program
$$\sup_{\theta \in P(\M)}  g(\theta).$$
We  use the framework of the ellipsoid method to approximately solve this convex program and  find a point which gives us a good enough estimate to  $g(\theta^\star).$
Note that, unlike the results in the previous section,  the bounding box here is easily obtained since $\theta \in P(\M) \subseteq [0,1]^m.$

\subsection{The Interior of $P(\M)$}
To check how close a candidate point $\theta$ is to $\theta^\star,$ we  use the max-entropy separation oracle provided to us.
The main difficulty we encounter is that the running time of the max-entropy oracle with marginals $\theta$ is inverse-polynomially dependent on the interiority of the point $\theta.$ Note that interiority of $\theta$ is a pre-requisite for strong duality to hold and for a succinct representation of the entropy-maximizing probability distribution to exist as in Lemma \ref{prop:maxentropya}.
The reason we assume a max-entropy oracle that works only if $\theta$ is in the inverse-polynomial interior of $P(\M)$ is that such an oracle is the best we can hope for algorithmically
 and, indeed, Theorems \ref{thm:strong} and \ref{thm:approximate-count} provide such an oracle. Without this restriction the proof of Theorem \ref{thm:reverse}  is simpler, but the theorem itself is  less useful as there may not exist a max-entropy oracle whose running time {\em does not} depend on the interiority of $\theta.$

\newcommand{\thetap}{{\theta^\bullet}}

The first issue raised by interiority is whether the point we are looking for, $\theta^\star$ may not be in the inverse-polynomial interior of $P(\M).$   To tackle this, we show that there is a point $\thetap$  in the $\eta$-interior of $P(\M)$  for $\eta ={\rm poly}(\eps, \nfrac{1}{m})$ such that  $g(\thetap) \sim g(\theta^\star)=\ln |\M|.$
Thus, instead of aiming for $\theta^\star,$ the ellipsoid algorithm  aims for $\thetap.$
\begin{lemma}\label{lem:interior-count}
Given an $\epsilon>0$, there exists an $\eta>0$ and $\thetap$ such that ${\thetap}$ is in $\eta$-interior of $P(\M)$ and $g(\thetap)\geq \left(1-\frac{\epsilon}{16m}\right)\ln |\M| \geq \ln |\M| -\frac{\epsilon}{16}$. Moreover, $\eta$ is at least a polynomial in $\nfrac{1}{m}$ and $\eps.$
\end{lemma}

\noindent
Before we prove this lemma, we show that  there exists {\em some} point $\tilde{\theta}$ in the ${\rm poly}(\nfrac{1}{m})$ interior of $P(\M).$ Such a point is  then used to show the existence of $\thetap.$
Note that we do not need to bound $g(\tilde{\theta})$ and only use the fact that it is non-negative.

\begin{lemma}[Same as Lemma \ref{lem:interiority0}]
\label{lem:interiority}
Let $\M \subseteq \{0,1\}^m$ and  $P(\M)=\left\{\vec{x}\in \RR^m_{\geq 0}: A^= \vec{x}=\vec{b}, \; A^{\leq}\vec{x}\leq \vec{c}\right\}$ be such that all the entries in $A^{\leq},\vec{c}\in \frac{1}{k_l} \cdot \ZZ$ and their  absolute values are at most $k_u.$
Then there exists a ${\tilde{\theta}}\in P(\M)$ such that $\tilde{\theta}$ is in the $\frac{1}{k_lk_u m^{1.5}}$-interior of $P(\M).$ Thus, if $k_u,k_l = {\rm poly}(m),$ then $\tilde{\theta}$ is in ${\rm poly}(\nfrac{1}{m})$-interior of $P(\M).$
\end{lemma}

\begin{proof}
Let $r$ be the dimension of $P(\M)$. Then, there exist $r+1$ affinely independent vertices $\vec{z}_0,\ldots,\vec{z}_r \in \{0,1\}^m$. We claim that  $\tilde{\theta} \defeq \frac{1}{r+1}\sum_{i=0}^r \vec{z}_i$ satisfies the conclusion of the lemma.
 Let  $$F_i \defeq \{x \in P(\M): A^{\leq}_i \vec{x}=   c_i\}$$ be a facet of $P(\M)$ where the  inequality constraint is one of  $(A^{\leq}, \vec{c}).$ Since the dimension of a facet is one less than that of a polytope, at least one of $z_0,z_1,\ldots, z_r,$  say $\vec{z}_0,$ does not lie in $F$ and, hence,  $A^\leq_i z_0<  c_i.$  Therefore,  $$A^{\leq}_i\vec{z}_0\leq c_i-\frac{1}{k_l}$$ since all coefficients are $\nfrac{1}{k_l}$-integral. Thus, the distance of $\vec{z}_0$ from $F$ is at least $$\frac{1}{k_l \|A^{\leq}_i\|}\geq \frac{1}{k_lk_u \sqrt{m}}.$$ Hence, the distance of $\tilde{\theta}$ from $F$ is at least $\frac{1}{m}\cdot \frac{1}{k_uk_l \sqrt{m}}.$ Since this argument works for any facet $F,$ the distance of $\tilde{\theta}$ from every facet of $P(\M)$ is at least $\frac{1}{m}\cdot \frac{1}{k_uk_l \sqrt{m}}.$
\end{proof}

\noindent
Henceforth, we assume that $k_l, k_u = {\rm poly}(m).$

\begin{proof}[Proof of Lemma \ref{lem:interior-count}.]
Let $\tilde{\theta}$ be the point in the interior of $P(\M)$ as guaranteed by Lemma \ref{lem:interiority}. Consider the point
$$\thetap \defeq \left(1-\frac{\epsilon}{16m^3}\right) \theta^\star + \frac{\epsilon}{16m^3} \tilde{\theta}.$$ Since $\tilde{\theta}$ is in the ${\rm poly}(\nfrac{1}{m}, \eps)$-interior of $P(\M),$  $\thetap$ must also be in the ${\rm poly}(\nfrac{1}{m}, \eps)$-interior of $P(\M).$
 On the other hand, since $g( \cdot)$ is a concave and non-negative function of $\vtheta$, we have that
$$g(\thetap)\geq  \left(1-\frac{\epsilon}{16m^3}\right) \cdot g(\theta^\star)\geq g\left(\theta^\star\right)-\frac{\epsilon}{16},$$ where we used the fact that $\ln |\M|\leq m$.

\end{proof}

\subsection{A Separation Oracle for Interiority}
Our final ingredient is a test for checking whether a point $\vtheta$ is in the inverse-polynomial interior of $P(\M).$  We show that the separation oracle for $P(\M)$ can be used to give such a test. We state the result in generality for any polyhedron $P$. For any $\eta >0$, let  $$P_\eta \defeq \{\vx : \vy\in P \;\; \forall \vy \textrm{ such that } \|\vx-\vy\| \leq \eta\}$$ denote the set of $\eta$-interior points in $P$.

\begin{lemma}\label{lem:test-interior}
There exists an algorithm that given a separation oracle  for a  polyhedron $P,$  a set of maximal linearly independent equalities $(A^=,b)$ satisfied by $P$, an $\eta>0$, and $\vtheta\in \RR^m$, either
\begin{enumerate}
\item asserts that $\vtheta \in P_{\nfrac{\eta}{2m}}$, i.e., $\vtheta$ is in the  $\nfrac{\eta}{2m}$-interior of $P$, or,
\item returns  $a$ such that  $\langle \va, \vy \rangle <\langle \va ,\vtheta \rangle $ for each $\vy\in P_{\eta}$, or equivalently, a separating hyperplane which separates the $\eta$-interior of $P$ from $\vtheta$.
\end{enumerate}
\end{lemma}
\begin{proof}
We use the the separation oracle for $P$ on a collection of a small number of  points close to $\vtheta$ to deduce  if $\vtheta$ is in the interior of $P$.  Even if one of these points is not in $P$, we use a separating hyperplane for such a point to separate
$\vtheta$ from the  interior of $P$. First, we describe the procedure when $P$ is full dimensional. Let $\vx_0,\ldots \vx_{m}$ form an $\eta$-regular simplex with center $\vtheta$, i.e.,
$$\frac{1}{m}\sum_{i=0}^m \vx_i=\vtheta \ \  {\rm and}  \ \ \|\vx_i-\vx_j\|= \eta$$ for each $i\neq j$. Such $x_0,\ldots, x_m$ can be found by starting with a regular simplex and then translating and scaling it. Now, the algorithm applies the separation oracle for each of $\vx_i$. Suppose the separation oracle asserts that $\vx_i \in P$ for each $0\leq i \leq m$. In this case, we assert that $\vtheta\in P_{\nfrac{\eta}{2m}}$. Observe that since each vertex of the simplex is in $P$, we must have that the whole simplex is in $P$. Since the simplex is regular where each edge is length $\eta$ and the center is $\vtheta$, there exists a ball of radius $\frac{\eta}{2m}$ centered at $\vtheta$ which is contained in the simplex and, hence, in $P$. Thus, $\vtheta$ is in $\frac{\eta}{2m}$-interior of $P$ as asserted.

Now, suppose that $\vx_i\notin P$ for some $i$ and let $\va$ be the separating hyperplane, i.e., $\langle \va, \vy \rangle <  \langle \va, x_i \rangle$ for each $\vy\in P$. Then consider the constraint $\langle \va, \vy \rangle <\langle \va ,\vtheta \rangle$. We claim that it is satisfied by each $\vy\in P_{\eta}$. Let $\vy\in P_{\eta}$ and consider
$$\vy' \defeq \vy -\vtheta+\vx_i.$$ Since $\|\vtheta-\vx_i\|\leq \eta$ and $\vy\in P_{\eta}$, we have that $\vy'\in P$ which implies that $\langle \va, \vy' \rangle  < \langle \va, \vx_i \rangle$. But this implies that
$$\langle \va,  \vy \rangle =\langle \va , \vy' -\vx_i+\vtheta \rangle =\langle \va, \vy' \rangle  - \langle \va,  \vx_i  \rangle+ \langle \va,\vtheta \rangle <\langle \va, \vtheta \rangle$$ which gives us the required separating hyperplane.

Now consider the case where $P$ is not full dimensional and let $r$ be the dimension of $P$. Recall that in this case we define interior of $P$ by restricting our attention to points in the affine space $\{x:A^=\vx=\vec{c}\}$. We modify the algorithm to chose a $r$-dimensional simplex in this affine space and check whether each of the vertices of the simplex is in $P$. The analysis is identical in this case.
\end{proof}

\subsection{The Ellipsoid Algorithm for Theorem \ref{thm:reverse}}

Now we present  the ellipsoid algorithm to approximately solve the convex program  $\min_{\vtheta\in P(\M)}g(\vtheta)$ and  prove Theorem~\ref{thm:reverse}. The starting ellipsoid is a ball of radius $\sqrt{m}$ that contains $[0,1]^m$ which contains $P(\M).$ Let us fix an $\epsilon>0$ and apply Lemma~\ref{lem:interior-count} to obtain  $\eta$ which is a polynomial in $\nfrac{1}{m}$ and ${\epsilon}$ and guarantees the existence of $\thetap$ in the $\eta$-interior such that $g(\thetap)\geq \ln |\M| -\frac{\epsilon}{16}$. In the range $(0,\ln |\M|]$
 we perform a binary search for the highest $\zeta$ such that the set $$S({\zeta},\eta) \defeq \{\vtheta\in P_{\eta} (\M):g(\vtheta)\geq\zeta\}$$ is non-empty when we search $\zeta$ within an accuracy $\nfrac{\epsilon}{16}$.

Given a  guess $\zeta$ for $g(\thetap),$ at an  iteration $t$ of the ellipsoid algorithm, we use the center $\theta_t$  of the ellipsoid as a guess for $\thetap.$ Ideally, we would pass $\theta_t$ to the max-entropy oracle which would either assert that $g(\theta_t) \geq \zeta - \nfrac{\eps}{16}$ or returns a $\lambda_t$ such that $f_{\theta_t}(\lambda_t) \leq \zeta +\nfrac{\eps}{16}.$ In the first case we  stop and return $\theta_t.$ In the latter case, we continue the search and use this $\lambda_t$ returned by the max-entropy oracle to update the ellipsoid into one with a smaller volume. However, to get the guarantee on the running time, we need to first check that the candidate point $\theta_t$ is in the $\eta$-interior of $P(\M).$ Here, we use the separation oracle from Lemma \ref{lem:test-interior}. We proceed to the max-entropy oracle only if this separation oracle asserts that  the point $\theta_t$ is in the $\nfrac{\eta}{2m}$-interior of $P(\M).$  In case this separation oracle outputs a hyperplane separating $\theta_t$ from $P_\eta(\M),$ we use this hyperplane to update the ellipsoid.
The key technical fact we show is that when $|\zeta - g(\thetap)| \leq \nfrac{\eps}{16},$ $\thetap$ is always contained in every ellipsoid.
Thus, once the radius of the ellipsoid becomes small enough,
we can output its center as a guess for $\thetap.$
Since the radius of the final ellipsoid is small and contains $\thetap$, the following lemma implies that the value of $g(\cdot)$ at the center of the ellipsoid is close enough to $g(\thetap)$ and, hence, by Lemma \ref{lem:interior-count}, to $g(\theta^\star).$

\begin{lemma}\label{lem:marginal_entropy}
Let $\vec{\theta},\vec{\theta}'\in P(\M)$ such that $\|\vec{\theta}-\vec{\theta}'\|\leq \epsilon$ and $\vec{\theta}$ is in $\eta$-interior of $P(\M)$. Then
$$g({\theta}')  \geq (1-\nfrac{\epsilon}{\eta})g(\theta).$$
\end{lemma}

\begin{proof}
Let ${\vec{\lambda}}^{\star}$ be an optimal solution to $\inf_{\vec{\lambda}}f_{\vec{\theta}}(\vec{\lambda})$. Thus, $p^{{\vlambda}^{\star}}$ is the optimal solution to primal convex program and $H(p^{{\vlambda}^{\star}})=\inf_{\vec{\lambda}}f_{\vec{\theta}}(\vec{\lambda})$. We  construct a probability distribution ${q}$ which is feasible for the primal convex program with parameter $\vec{\theta}'$ and $H({q})\geq (1-\nfrac{\epsilon}{\eta})H(p^{{\vlambda}^{\star}})$, thus, proving the lemma.
We begin with a claim.
\begin{claim}
Let $\vec{\theta}'' \defeq \frac{\vec{\theta}'-(1-\nfrac{\epsilon}{\eta})\vec{\theta}}{\nfrac{\epsilon}{\eta}}$. Then $\vec{\theta}''\in P(\M)$.
\end{claim}
\begin{proof}
First, observe that any equality constraint for $P(\M)$ of the form $\langle A^=_i,\vec{x}\rangle=b_i$ is satisfied by both $\vec{\theta}$ and $\vec{\theta}'$. Therefore,
$$\langle A^=_i,\vec{\theta}''\rangle=\frac{\langle A^=_i,\vec{\theta}'\rangle-(1-\nfrac{\epsilon}{\eta})\langle A^=_i,\vec{\theta}\rangle}{\nfrac{\epsilon}{\eta}}=\frac{b_i-(1-\nfrac{\epsilon}{\eta})b_i}{\nfrac{\epsilon}{\eta}}=b_i.$$
Thus, it is enough to show that $\|\vec{\theta}''-\vec{\theta}\|\leq \eta.$ To see this note that
\begin{align*}
\|\vec{\theta}''-\vec{\theta}\|=\left\|\frac{\vec{\theta}'-(1-\nfrac{\epsilon}{\eta})\vec{\theta}}{\nfrac{\epsilon}{\eta}}-\vec{\theta}\right\|=\left\|\frac{\vec{\theta}'-\vec{\theta}}{\nfrac{\epsilon}{\eta}}\right\|\leq \frac{\epsilon}{\nfrac{\epsilon}{\eta}}=\eta
\end{align*}
where the last inequality follows from the fact that $\|\vec{\theta}-\vec{\theta}'\|\leq \epsilon$.
\end{proof}

\noindent
Let ${q}''$ be an arbitrary probability measure over $\M$ such that the marginals of ${q}''$ equal $\vec{\theta}''$, i.e., $\theta''_e=\sum_{M\in \M: e\in M}{q''}_M$.
Let ${q}$ be the probability measure defined to be
$$q \defeq (1-\nfrac{\epsilon}{\eta})p^{{\vlambda}^{\star}}+\nfrac{\epsilon}{\eta}q''.$$ Then $$\sum_{M\in \M:e\in M}q_M=(1-\nfrac{\epsilon}{\eta})\theta_e+\nfrac{\epsilon}{\eta}\theta''_e=\theta'_e.$$ By concavity and non-negativity of the entropy function, we have
$$H(q)\geq (1-\nfrac{\epsilon}{\eta})H(p^{{\vlambda}^{\star}})+\nfrac{\epsilon}{\eta}H(q'')\geq (1-\nfrac{\epsilon}{\eta})H(p^{{\vlambda}^{\star}})$$
as required.
\end{proof}

\noindent
We now move on to the description of the ellipsoid algorithm and subsequently complete the proof of Theorem \ref{thm:reverse}.

\begin{enumerate}
\item {\bf Input}
\begin{enumerate}
\item An error parameter $\eps>0.$
\item A maximally linearly independent set of equalities $(A^=,b)$ for $P(\M).$
\item A separation oracle for the facets $(A^\leq,c)$ of $P( \M).$
\item A max-entropy oracle for $\M.$
\item A guess $\zeta \in (0,m]$ (for  $g(\thetap)$).
\end{enumerate}

\item {\bf Initialization}
\begin{enumerate}
\item Let $\eta$ be as in Lemma \ref{lem:interior-count}.
\item  Let  $E_0 \defeq E(B_{0},\vec{c}_{0})$ be a sphere with radius $R={\sqrt{m}}$ containing $[0,1]^m$ restricted to  the affine space $A^=x=b.$
\item Set $t=0$
\end{enumerate}

\item {\bf Repeat} until  the ellipsoid $E_t$ is contained in a ball of radius at most $\frac{\epsilon\eta}{16m}.$
\begin{enumerate}
\item Given the ellipsoid  $E_{t} \defeq E(B_{t},\vec{c}_{t})$, set  $\vtheta_t \defeq\vec{c}_{t}$.

\item Check using the separation oracle for $P(\M)$ as in Lemma~\ref{lem:test-interior}  {\bf if}  $\vtheta_t\in P_{\frac{\eta}{2m}}(\M)$
\begin{enumerate}
\item {\bf then} {\bf goto} Step (\ref{rev-call-oracle})

\item\label{hyperplane-sep} {\bf else} let  $e_t$ be the separating hyperplane returned as in Lemma \ref{lem:test-interior}, i.e., $\langle e_t,\vtheta -\vtheta_t\rangle\geq 0$ for all $\vtheta \in P_{\eta}(\M)$ and {\bf goto} Step (\ref{rev-update-ellipse}).
\end{enumerate}

\item\label{rev-call-oracle} Call the max-entropy oracle  with input $\vtheta_t$, $\zeta$ and $\nfrac{\epsilon}{16}$.

\item  {\bf If} $g(\vtheta_t)\geq \zeta-\nfrac{\epsilon}{16}$
\begin{enumerate}
\item {\bf then} {\bf return}  $\vtheta_t$ and {\bf stop}.
\item\label{hyperplane-entropy}  {\bf else}  the max-entropy oracle returns $\vlambda_t$ such that
$f_{\vtheta_t}(\vlambda_t) \leq \zeta+\nfrac{\epsilon}{16}$. Let $e_t=\vlambda_t$.
\end{enumerate}

 \item\label{rev-update-ellipse} Compute the ellipsoid $E_{t+1}$ to be the smallest ellipsoid containing the half-ellipsoid
 $\{\vtheta \in E_t: \langle e_t,\vtheta-\vtheta_t \rangle \geq 0\}$ and restricted to  the affine space $A^=x=b.$

\item $t=t+1$ .
\end{enumerate}

\item Let $T=t$ call the max-entropy oracle with input $\theta_T, \zeta$ and $\nfrac{\eps}{16}.$

\item {\bf If}  $g(\theta_T) \geq \zeta  - \nfrac{\eps}{16}$
\begin{enumerate}
 \item  {\bf then}  {\bf return} $\theta_T$ and {\bf stop}.
\item  {\bf else} {\bf return} $g(\thetap) < \zeta$ and {\bf stop}  ($\zeta$ is not a good guess for $g(\thetap)$).
\end{enumerate}
\end{enumerate}

\noindent
It is clear that any call to the approximate optimization oracle is made for points $\vtheta$ which are in $\frac{\eta}{2m}$ interior. Thus, the running time of the algorithm is polynomially bounded by $m$ and $\nfrac{1}{\eps}$ for each call. To bound the number of iterations note that the starting ellipsoid has radius $\sqrt{m}$ and the final ellipsoid ${\rm poly}(\nfrac{1}{m},\eps).$ Hence, the number of iterations can be bounded by  Theorem \ref{thm:john} by ${\rm poly}(m, \nfrac{1}{\eps}).$ It remains to prove the correctness of the algorithm.

Towards this, let ${\zeta^\circ}$ be the largest guess of $\zeta$ for which the algorithms returns a positive answer and let ${\vtheta^\circ}$ be the point returned by the algorithm for guess $\zeta^\circ$. We return $Z^\circ\defeq e^{{\zeta^\circ}}$ as our estimate of $|\M|$.
To complete the proof of Theorme \ref{thm:reverse}, we show  that $Z^\circ$ satisfies
\begin{equation}\label{main-reverse}
(1-\epsilon)|\M|\leq Z^\circ\leq (1+\epsilon)|\M|.
\end{equation}

\noindent
First, we prove the following lemma.
\begin{lemma}\label{cl:survive-reverse}
Consider the run of the ellipsoid algorithm for a guess $\zeta$ and let the hyperplane $\{\vtheta: \langle e_t,\vtheta-\vtheta_t \rangle \geq 0\}$ be used as a separating hyperplane in some iteration of the algorithm. Then this separating hyperplane does not cut any point $\theta$ such that $\theta \in P_{\eta}(\M)$ and $g(\theta) \geq \zeta+ \nfrac{\eps}{16}.$
\end{lemma}
\begin{proof}
If the hyperplane $e_t$ is obtained in Step (\ref{hyperplane-sep}), then it is clearly a valid inequality for $P_{\eta}(\M)$ and therefore does not cut off any of its points. Otherwise, suppose $e_t=\vlambda_t$ is obtained in Step (\ref{hyperplane-entropy}). Then
$$f_{\vtheta_t}(\vlambda_t)=\langle \vlambda_t,\vtheta_t \rangle+ \ln{ Z^{\vlambda_t}} \leq \zeta +\frac{\epsilon}{16}.$$ Hence,
\begin{eqnarray}\label{eq-3}
\langle \vlambda_t,\vtheta-\vtheta_t \rangle=f_{\vtheta}(\vlambda_t)-f_{\vtheta_t}(\vlambda_t) \geq f_{\vtheta}(\vlambda_{t})-\zeta-\frac{\epsilon}{16}.
\end{eqnarray}
Thus, by the assumption in the lemma,  $$f_{\vtheta}(\vlambda_t)\geq g(\theta) \geq \zeta+\nfrac{\epsilon}{16}$$ and, therefore, by \eqref{eq-3}, $\vtheta$ satisfies the constraint $\langle e_t,\vtheta-\vtheta_t\rangle\geq 0$.
\end{proof}

\noindent
We now show that $\zeta^\circ\geq \ln|\M|-\frac{4\epsilon}{16}$. Consider the run of the algorithm for
$$\zeta'\in \left[\ln|\M|-\frac{4\epsilon}{16},\ln|\M|-\frac{3\epsilon}{16}\right].$$ Since
$$g(\thetap)\geq \ln|\M|-\frac{\epsilon}{16}\geq \zeta'+\frac{\epsilon}{16},$$ $\thetap$ cannot be cut off in any iteration by Lemma~\ref{cl:survive-reverse}. If the ellipsoid returns an answer when run with guess $\zeta'$ then
$$\zeta^\circ\geq \zeta'\geq \ln|\M|-\frac{4\epsilon}{16}$$ as claimed. Otherwise, we end with an ellipsoid $E_T$ of radius at most $\frac{\epsilon\eta}{16m}$. Let $\theta_T$ be the center of the ellipsoid $E_T$.  Since $\thetap\in E_T$, we have that $\|\thetap-\theta_T\|\leq \frac{\epsilon\eta}{16 n}$. Since $\thetap$ is in $\eta$-interior, from Lemma~\ref{lem:marginal_entropy} it follows that
\begin{eqnarray*}
g(\theta_T)\geq \left(1-\frac{\frac{\epsilon\eta}{16m}}{\eta}\right)g(\thetap)\geq \left(1-\frac{\epsilon}{16m}\right)g(\thetap)> \ln |\M|-\frac{2\epsilon}{16}.
\end{eqnarray*}
\noindent
This contradicts the fact that the algorithm did not output $\vtheta_T$ in the last iteration and asserted
$$g(\theta_T)\leq  \zeta'+ \frac{\epsilon}{16} \leq  \ln |\M|-\frac{2\epsilon}{16}.$$
Since
$\zeta^\circ\geq \ln |\M|-\frac{4\epsilon}{16}$ and $\zeta^\circ\leq \ln |\M|+\frac{\epsilon}{16}$. We obtain that $(1-\epsilon)|\M|\leq Z^\circ\leq (1+\epsilon)|\M|$ proving \eqref{main-reverse} and completing the proof of Theorem \ref{thm:reverse}.

\bibliographystyle{abbrv}
\bibliography{entropy}

\begin{thebibliography}{10}

\bibitem{AsadpourGMGS10}
A.~Asadpour, M.~X. Goemans, A.~Madry, S.~Oveis-Gharan, and A.~Saberi.
\newblock An o(log n/ log log n)-approximation algorithm for the asymmetric
  traveling salesman problem.
\newblock In {\em SODA}, pages 379--389, 2010.

\bibitem{AsadpourS10}
A.~Asadpour and A.~Saberi.
\newblock An approximation algorithm for max-min fair allocation of indivisible
  goods.
\newblock {\em SIAM J. Comput.}, 39(7):2970--2989, 2010.

\bibitem{BentalN12}
A.~Ben-Tal and A.~Nemirovski.
\newblock Optimization {III}: Convex analysis, nonlinear programming theory,
  nonlinear programming algorithms.
\newblock Lecture Notes, 2012.

\bibitem{BoydV04}
S.~Boyd and L.~Vandenberghe.
\newblock {\em {Convex Optimization}}.
\newblock Cambridge University Press, Mar. 2004.

\bibitem{KL}
T.~M. Cover and J.~Thomas.
\newblock {\em Elements of information theory}.
\newblock Wiley, New York, 1991.

\bibitem{DyerFKKPV93}
M.~E. Dyer, A.~M. Frieze, R.~Kannan, A.~Kapoor, L.~Perkovic, and U.~V.
  Vazirani.
\newblock A mildly exponential time algorithm for approximating the number of
  solutions to a multidimensional knapsack problem.
\newblock {\em Combinatorics, Probability {\&} Computing}, 2:271--284, 1993.

\bibitem{Edmonds65b}
J.~Edmonds.
\newblock Maximum matching and a polyhedron with $0,1$ vertices.
\newblock {\em Journal of Research of the National Bureau of Standards},
  69:125--130, 1965.

\bibitem{edmonds1970submodular}
J.~Edmonds.
\newblock Submodular functions, matroids, and certain polyhedra.
\newblock {\em Combinatorial structures and their applications}, pages 69--87,
  1970.

\bibitem{Edmonds71}
J.~Edmonds.
\newblock Matroids and the greedy algorithm.
\newblock {\em Mathematical Programming}, 1:127--136, 1971.
\newblock 10.1007/BF01584082.

\bibitem{Elbassioni201256}
K.~Elbassioni and H.~R. Tiwary.
\newblock Complexity of approximating the vertex centroid of a polyhedron.
\newblock {\em Theoretical Computer Science}, 421(0):56 -- 61, 2012.

\bibitem{FederM92}
T.~Feder and M.~Mihail.
\newblock Balanced matroids.
\newblock In S.~R. Kosaraju, M.~Fellows, A.~Wigderson, and J.~A. Ellis,
  editors, {\em STOC}, pages 26--38. ACM, 1992.

\bibitem{frieze1982worst}
A.~M. Frieze, G.~Galbiati, and F.~Maffioli.
\newblock On the worst-case performance of some algorithms for the asymmetric
  traveling salesman problem.
\newblock {\em Networks}, 12(1):23--39, 1982.

\bibitem{Birkhoff46}
B.~Garrett.
\newblock Three observations on linear algebra. (spanish).
\newblock {\em Univ. Nac. Tucum�n. Revista A.}, 5:147--151, 1946.

\bibitem{GodsilR01}
C.~D. Godsil and G.~Royle.
\newblock {\em Algebraic Graph Theory}.
\newblock Springer, 2001.

\bibitem{GrotschelLS88}
M.~Gr{\"o}tschel, L.~Lov{\'a}sz, and A.~Schrijver.
\newblock {\em {Geometric Algorithms and Combinatorial Optimization}}, volume~2
  of {\em Algorithms and Combinatorics}.
\newblock Springer, second corrected edition edition, 1993.

\bibitem{HuangK12}
Z.~Huang and S.~Kannan.
\newblock The exponential mechanism for social welfare: Private, truthful, and
  nearly optimal.
\newblock In {\em FOCS}, pages 140--149, 2012.

\bibitem{Jaynes1}
E.~T. {Jaynes}.
\newblock {Information Theory and Statistical Mechanics}.
\newblock {\em Physical Review}, 106:620--630, May 1957.

\bibitem{Jaynes2}
E.~T. {Jaynes}.
\newblock {Information Theory and Statistical Mechanics. II}.
\newblock {\em Physical Review}, 108:171--190, Oct. 1957.

\bibitem{Jerrum95}
M.~Jerrum.
\newblock A very simple algorithm for estimating the number of k-colorings of a
  low-degree graph.
\newblock {\em Random Struct. Algorithms}, 7(2):157--166, 1995.

\bibitem{JerrumS89}
M.~Jerrum and A.~Sinclair.
\newblock Approximating the permanent.
\newblock {\em SIAM J. Comput.}, 18(6):1149--1178, 1989.

\bibitem{JerrumSV04}
M.~Jerrum, A.~Sinclair, and E.~Vigoda.
\newblock A polynomial-time approximation algorithm for the permanent of a
  matrix with nonnegative entries.
\newblock {\em J. ACM}, 51(4):671--697, 2004.

\bibitem{JerrumVV86}
M.~Jerrum, L.~G. Valiant, and V.~V. Vazirani.
\newblock Random generation of combinatorial structures from a uniform
  distribution.
\newblock {\em Theor. Comput. Sci.}, 43:169--188, 1986.

\bibitem{Kahn96}
J.~Kahn.
\newblock Asymptotics of the chromatic index for multigraphs.
\newblock {\em J. Comb. Theory, Ser. B}, 68(2):233--254, 1996.

\bibitem{Kahn00}
J.~Kahn.
\newblock Asymptotics of the list-chromatic index for multigraphs.
\newblock {\em Random Struct. Algorithms}, 17(2):117--156, 2000.

\bibitem{KahnK97}
J.~Kahn and P.~M. Kayll.
\newblock On the stochastic independence properties of hard-core distributions.
\newblock {\em Combinatorica}, 17(3):369--391, 1997.

\bibitem{Kapur}
J.~N. Kapur.
\newblock {\em Maximum-Entropy Models in Science and Engineering}.
\newblock Wiley, New York, 1989.

\bibitem{KarpL83}
R.~M. Karp and M.~Luby.
\newblock Monte-carlo algorithms for enumeration and reliability problems.
\newblock In {\em FOCS}, pages 56--64. IEEE Computer Society, 1983.

\bibitem{Kirchoff47}
G.~Kirchhoff.
\newblock Ueber die au\"osung der gleichungen, auf welche man bei der
  untersuchung der linearen vertheilung galvanischer str\"ome gef\"uhrt wird.
\newblock {\em Ann. Phys. und Chem.}, 72:497�508, 1847.

\bibitem{Nesterov}
Y.~Nesterov.
\newblock Introductory lectures on convex programming.
\newblock Volume {I}: Basic course, 1998.

\bibitem{GharanSS11}
S.~Oveis-Gharan, A.~Saberi, and M.~Singh.
\newblock A randomized rounding approach to the traveling salesman problem.
\newblock In R.~Ostrovsky, editor, {\em FOCS}, pages 550--559. IEEE, 2011.

\bibitem{padberg1982odd}
M.~W. Padberg and M.~R. Rao.
\newblock Odd minimum cut-sets and $b$-matchings.
\newblock {\em Mathematics of Operations Research}, 7(1):67--80, 1982.

\bibitem{JerrumSinclair}
A.~Sinclair and M.~Jerrum.
\newblock Approximate counting, uniform generation and rapidly mixing {M}arkov
  chains.
\newblock {\em Inf. Comput.}, 82(1):93--133, July 1989.

\bibitem{StockMeyer85}
L.~J. Stockmeyer.
\newblock On approximation algorithms for $\#$p.
\newblock {\em SIAM J. Comput.}, 14(4):849--861, 1985.

\bibitem{Valiant79a}
L.~G. Valiant.
\newblock The complexity of computing the permanent.
\newblock {\em Theor. Comput. Sci.}, 8:189--201, 1979.

\bibitem{Valiant79b}
L.~G. Valiant.
\newblock The complexity of enumeration and reliability problems.
\newblock {\em SIAM J. Comput.}, 8(3):410--421, 1979.

\bibitem{Vishnoi12}
N.~K. Vishnoi.
\newblock A permanent approach to the traveling salesman problem.
\newblock In {\em FOCS}, pages 76--80, 2012.

\bibitem{WainwrightJ08}
M.~J. Wainwright and M.~I. Jordan.
\newblock Graphical models, exponential families, and variational inference.
\newblock {\em Foundations and Trends in Machine Learning}, 1(1-2):1--305,
  2008.

\end{thebibliography}

\appendix

\newpage

\section{Omitted Proofs}\label{sec:omit}

\subsection{Duality of the Max-Entropy Program}\label{sec:omit-duality}
\begin{lemma}\label{prop:maxentropy}
For a point $\vec{\theta}$ in the interior of $P(\M),$  there exists a unique distribution $p^\star$ which attains the max-entropy while satisfying
$$\sum_{M \in \M} p^\star_M \one_M=\vec{\theta}.$$ Moreover, there exists $\lambda^\star :[m] \mapsto \mathbb{R}$ such that $p^\star_M\propto e^{-\lambda^\star(M)}$ for each $M\in \M$.
\end{lemma}

\begin{proof}
Consider the convex program for computing the maximum-entropy distribution with marginals $\vec{\theta}$ as in Figure \ref{conv-entropy}.
\begin{figure}[h]
\begin{center}
\begin{eqnarray}
  \sup & \sum_{M\in \M} p_M \ln\frac{1}{p_M}  \nonumber\\
\st \nonumber \\
\forall e\in [m]& \sum_{M\in \M, \; M \ni e} p_M=\theta_e \label{cons1a}\\
  & \sum_{M\in \M} p_M =1\label{cons2a}\\
\forall M\in \M & p_M \geq 0\label{cons3a}
\end{eqnarray}
\end{center}
\caption{Max-Entropy Program for ($\M,\vec{\theta}$)} \label{conv-entropy}
\end{figure}
We first prove that the dual of this convex program is the one given in Figure \ref{conv-dual-entropy}.
\begin{figure}[h]
\begin{center}
\begin{eqnarray}
\inf & f_\theta(\vec{\lambda})\defeq \sum_{e\in [m]}\theta_e \lambda_e+\ln \sum_{M\in \M}e^{-\lambda(M)}  \nonumber\\
\st \nonumber \\
\forall e\in [m] & \lambda_e \in \mathbb{R} \label{cons4a}
\end{eqnarray}
\end{center}
\caption{Dual of the Max-Entropy Program for ($\M,\vec{\theta}$)} \label{conv-dual-entropy}
\end{figure}
To see this consider multipliers $\lambda_e$ for constraints \eqref{cons1a} in Figure \ref{conv-entropy} and a multiplier $z$ for the constraint \eqref{cons2a}. Then the Lagrangian $L(p,\vec{\lambda},z)$ is defined to be
$$  \sum_{M \in \M}p_M \ln\frac{1}{p_M} + \sum_{e \in [m]} \lambda_e(\theta_e- \sum_{M\in \M, \; M \ni e} p_M) + z (1-\sum_{M\in \M} p_M ).$$
This is the same as
\begin{equation}\label{2eq2}
 \sum_{M \in \M}p_M \ln\frac{1}{p_M} - \sum_{M \in \M} p_M \lambda(M) - z \sum_{M\in \M} p_M  +  \sum_{e \in [m]} \lambda_e\theta_e +  z .
 \end{equation}
Let $g(\vec{\lambda},z) \defeq \inf _{p \geq 0} L(p,\vec{\lambda},z).$
Thus, the $p$ which achieves $g(\vec{\lambda},z)$ can be obtained by taking partial derivatives with respect to $p_M$ and setting them to $0$ as follows.
\begin{equation}\label{derivative}
 \forall {M \in \M} \quad \frac{\partial L}{\partial p_M}=  \ln\frac{1}{p_M} -1 -\lambda(M)-z=0.
 \end{equation}
Thus,  $p_M = e^{-1-z-\lambda(M)}$ for all $M \in \M.$
Summing this up over all $M \in \M$ we obtain that
\begin{equation}\label{2eq4}
 \sum_{M \in \M} p_M = e^{-1-z} \sum_{M \in \M} e^{-\lambda(M)}.\end{equation}
For  such a $(p,\vec{\lambda},z),$ if we multiply each \eqref{derivative} by $p_M$ and add all of them up we obtain
$$ \sum_{M \in \M} \left( p_M\ln\frac{1}{p_M} -p_M- p_M\lambda(M) -zp_M\right) = 0,$$ which implies that
$$
 \sum_{M \in \M} \left( p_M\ln\frac{1}{p_M} - p_M\lambda(M) -zp_M\right) =  \sum_{M \in \M} p_M.
$$
Hence, combining this with \eqref{2eq2} and using \eqref{2eq4}, the dual becomes to find the infimum of $g(\vlambda,z)$ which is
$$
\sum_{e \in [m]} \lambda_e \theta_e +z+ e^{-1-z} \sum_{M \in \M} e^{-\lambda(M)}.
$$
Optimizing $g(\vec{\lambda},z)$  over $z$ one obtains that $g(\vec{\lambda},z)$ is minimized when
$$ 1- e^{-1-z} \sum_{M \in \M} e^{-\lambda(M)}=0.$$
Hence, $z= \ln \sum_{M \in \M} e^{-\lambda(M)}-1.$
Thus, the Lagrangian dual becomes to minimize
$$\sum_{e\in [m]}\theta_e \lambda_e+\ln \sum_{M\in \M}e^{-\lambda(M)}.$$
This completes the proof that the dual of Figure \ref{conv-dual-entropy} is the convex program in Figure \ref{conv-entropy}.

Since $\vec{\theta}$ is in the interior of $P(\M),$ the primal-dual pair satisfies Slater's condition and  strong duality holds, see \cite{BoydV04}, implying that the optimum of both the programs is the same. Moreover, by the strict concavity of the entropy function, the optimum is unique. Hence, at optimality,  $p^\star_M=\frac{ e^{-\lambda^{\star}(M)}}{\sum_{N\in \M} e^{-\lambda^{\star}(N)}}$ where $\vec{\lambda}^{\star}$ is the optimal dual solution and $p^\star$ is the optimal primal solution.
\end{proof}

\subsection{Optimal and Near-Optimal Dual Solutions}\label{sec:omit-near-optimal}
In this section we first prove that if $\vec{\lambda}$ is a  solution to the program in Figure  \ref{conv-dual-entropy} for $(\M,\vec{\theta})$ of value $\zeta,$  then so is any $\vec{\lambda} + (A^=)^\top d$ for any $\vec{d}.$  Recall that $(A^=,\vec{b})$ are the equality constraints satisfied by all vertices of $\M.$ Hence, in our search for the optimal solution to the dual convex program, we  restrict ourselves to the space of $\vec{\lambda}$ s.t. $A^=\vec{\lambda}=0.$

\begin{lemma}[Same as Lemma \ref{lem:shift0}]\label{lem:shift}
$f_\theta(\vec{\lambda})= f_\theta(\vec{\lambda}+ (A^=)^\top d )$ for any $\vec{d}.$
\end{lemma}
\begin{proof}
First, note that $\langle \vec{\lambda} +(A^=)^\top d , \vec{\theta} \rangle =\langle \vec{\lambda}, \vec{\theta} \rangle +  \langle(A^=)^\top d , \vec{\theta} \rangle.$
Note that $\vec{\theta}$ can be written as $\sum_{M \in \M} p_M \one_M$ and $A^= \one_M=b$ for all $M \in  \M.$ Hence,
$$ \langle \vec{\lambda} + (A^=)^\top d , \vec{\theta} \rangle = \langle \vec{\lambda}, \vec{\theta} \rangle + \sum_{M \in \M} p_M \langle \vec{d},b \rangle =  \langle \vec{\lambda}, \vec{\theta} \rangle +  \langle \vec{d},b \rangle$$
since $\sum_{M \in \M} p_M =1.$
On the other hand note that
$$\ln \sum_{M \in \M} e^{-\langle \vec{\lambda} +(A^=)^\top d  , \one_M \rangle}=\ln e^{- \langle \vec{d},b \rangle} \sum_{M \in \M} e^{-\langle \vec{\lambda} , \one_M \rangle} $$
which equals
$$ - \langle \vec{d},b \rangle + \ln  \sum_{M \in \M} e^{-\langle \vec{\lambda} , \one_M \rangle}= - \langle \vec{d},b \rangle+ \ln  \sum_{M \in \M} e^{-{\lambda}(M)}.$$
Combining, we obtain that $f_\theta(\vec{\lambda} +(A^=)^\top d )$ equals
$$ \langle \vec{\lambda} + (A^=)^\top d , \vec{\theta} \rangle + \ln \sum_{M \in \M} e^{-\langle \vec{\lambda} +(A^=)^\top d  , \one_M \rangle} =   \langle \vec{\lambda}, \vec{\theta} \rangle  + \ln  \sum_{M \in \M} e^{-\lambda(M)}$$
which equals $f_\theta(\vec{\lambda}).$ This completes the proof of the lemma.
\end{proof}

\noindent
Thus, we can assume that $A^=\vec{\lambda}^{\star}=0$ where $\vec{\lambda}^{\star}$ is the optimal solution for the program of Figure \ref{conv-dual-entropy} for  $(\M,\theta).$

Next we prove that if $\vec{\lambda}$ is such that $f_\theta(\vec{\lambda})$ is close to $f_\theta(\vec{\lambda}^\star),$ then $p^{\vlambda}$ and $p^{\vlambda^\star}$ are close to each other. We relate the Kullback-Leibler distance between $p^{\vlambda}$ and $p^{\vlambda^\star}$ to $f_\theta(\vec{\lambda})-f_\theta(\vec{\lambda}^\star).$ In particular $\vec{\theta}^\lambda$ and $\vec{\theta}$ are close to each other.
Before we state this lemma, we recall some basic measures of proximity between probability distributions.

\begin{definition} Let $p,q$ be two probability distributions over the same space $\Omega.$  The following are natural measures of distances.

\begin{enumerate}
\item $\|p-q\|_{\rm TV} \defeq \max_{S \subseteq \Omega} |p(S)-q(S)|.$
\item $\|p-q\|_{1} \defeq \sum_{\omega} |p(\omega)-q(\omega)|.$
\item If $p,q>0,$ then the Kullback-Leibler distance between them is defined to be $$D_{\rm KL}(p \| q) \defeq \sum_{\omega} p(\omega) \ln \frac{p(\omega)}{q(\omega)}.$$ This distance function is always non-negative but not necessarily symmetric. 
\end{enumerate}
\end{definition}

\noindent
The following lemma shows a close relation between the dual solutions and Kullback-Leibler distance between the corresponding primal distributions.
\begin{lemma}\label{lem:close}
Suppose $\vec{\lambda}$ is such that $f_\theta(\vec{\lambda}) \leq f_\theta(\vec{\lambda}^\star)+\eps$  where $\vec{\lambda}^\star$ is the optimum dual solution for the instance $(\M,\vec{\theta}).$ Let   $p^{\vlambda},p^{\vlambda^\star}$ be the probability distributions corresponding to $\vec{\lambda}$ and $\vec{\lambda}^\star$ respectively: $$p_M^{\vlambda}  \defeq  \frac{e^{-\lambda(M)}}{\sum_{N \in \M} e^{-\lambda(N)}} \quad \mbox{and} \quad p^{\vlambda^\star}_M \defeq  \frac{e^{-\lambda^\star(M)}}{\sum_{N \in \M} e^{-\lambda^\star(N)}}.$$
Then
$$ f_\theta(\vec{\lambda})-f_\theta(\vec{\lambda}^\star)= D_{\rm KL}(p^{\vlambda^\star}\| p^{\vlambda})=\eps.$$

\end{lemma}

\begin{proof}
Let $Z^{\vlambda},Z^{\vlambda^\star}$ denote $\sum_{M \in \M} e^{-\lambda(M)}$ and $\sum_{M \in \M} e^{-\lambda^\star(M)}$ respectively. Then, it follows from optimality of $\vec{\lambda}^\star$ that
\begin{equation}\label{2eq3}
\theta_e = \frac{\sum_{M \in \M, \; M \ni e} e^{-\lambda^\star(M)} }{Z^{\vlambda^\star}}.
\end{equation}
Hence,
$$D_{\rm KL}(p^\star \| p) = \sum_{M \in \M} p_M^\star  \ln \frac{1}{p_M} - \sum_{M \in \M} p_M^\star  \ln \frac{1}{p_M^\star}$$
 which, since $p_M=\nfrac{e^{-\lambda(M)}}{Z^{\vlambda}}$ is
$$ \sum_{M \in \M} p_M^\star \ln Z^{\vlambda} + \sum_{M \in \M} p_M^\star \lambda(M) - \sum_{M \in \M} p_M^\star  \ln \frac{1}{p_M^\star}.$$
This is equal to
$$\ln Z^{\vlambda} - f_\theta(\vec{\lambda}^\star) +  \sum_{e\in [m]} \lambda_e \sum_{M \in \M, \; M \ni e} p_M^\star $$
$$= \ln Z^{\vlambda} - f_\theta(\vec{\lambda}^\star) + \langle \vec{\lambda}, \vec{\theta} \rangle  = f_\theta(\vec{\lambda})- f_\theta(\vec{\lambda}^\star).$$
Here, we have used \eqref{2eq3}.
Hence, $D_{\rm KL}(p^\star \|p) = f_\theta(\vec{\lambda})-f_\theta(\vec{\lambda}^\star)\leq \eps.$
\end{proof}

\noindent
It is well-known, see \cite{KL} Lemma 12.6.1, pp. 300-301,  that for probability distributions $p,q$ over the same sample space
\begin{equation}\label{eq:KL}
  \|p-q\|_{\rm TV} \leq O\left(\sqrt{{D_{\rm KL}(p \| q)}}\right).
\end{equation}
Hence, we obtain the following as a corollary to  Lemma \ref{lem:close}.

\begin{corollary}\label{cor:marginal}
Let $\vec{\lambda}$ be such that $f_\theta(\vec{\lambda}) \leq f_\theta(\vec{\lambda}^\star)+\eps.$ Then,
for all $e \in [m],$ $|{\theta}^{\vlambda}_e-\theta^{\vlambda^\star}_e| \leq O(\sqrt{\eps}).$
\end{corollary}

\begin{proof}
This follows from the fact that
$$|\theta^{\vlambda}_e-\theta^{\vlambda^\star}_e| \leq  \|p^{\vlambda}-p^{{\vlambda}^\star}\|_{\rm TV}$$
which is at most
$$ O\left(\sqrt{D_{\rm KL}(p^{{\vlambda}^\star} \| p^{\vlambda})}\right)= O\left(\sqrt{f_\theta(\vec{\lambda})-f_\theta(\vec{\lambda}^\star)}\right)=O(\sqrt{\eps}).$$
\end{proof}

\section{Generalized Counting and Minimizing Kullback-Leibler Divergence}\label{sec:general}

In this section, we outline how to obtain algorithms for generalized approximate counting from max-entropy oracles. 
Recall that the generalized approximate counting problem is: 
Given $\epsilon>0$ and  weights  $\vec{\mu}\in \mathbb{R}^{m}$,  output $\widetilde{Z}^{\mu}$ and $\widetilde{Z}^{\mu}_e$ for each $e\in [m]$ such that the  following guarantees hold.
\begin{enumerate}
\item $(1-\epsilon) Z^\mu \leq \widetilde{Z}^{{\mu}} \leq (1+\epsilon)  Z^\mu$ and
\item for every $e \in [m],$
$(1-\eps) Z^\mu_e \leq \widetilde{Z}^{\mu}_e \leq (1+\eps) Z^\lambda_e.$
\end{enumerate}
Here $Z^{\mu}$ is defined to be  $\sum_{M\in \M} e^{-\mu(M)}.$
The running time should be a  polynomial in $m$,  $\nfrac{1}{\epsilon},$  $\log {\nfrac{1}{\alpha}}$ and the number of bits required to represent $e^{-\mu_e}$ for any $e\in [m],$ or $\|\mu\|_1.$\footnote{The number of bits needed to represent $e^{-\mu_e}$ for any $e\in [m],$  up to an additive error of $2^{-m^2},$ is  $\max\{{\rm poly}(m),\|\mu\|_1\}$ for some ${\rm poly}(m).$ Since all our running times depend on   ${\rm poly}(m),$ we  only track the dependence on $\|\mu\|_1.$}
Towards constructing such algorithms, we need access to  oracles that solve a more general problem than the max-entropy problem used in Theorem~\ref{thm:reverse}, namely, min-Kullback-Leibler (KL)-divergence problem. This raises the issue that, while  Theorems \ref{thm:strong} and \ref{thm:approximate-count}   output solutions to max-entropy problem given access to generalized counting oracles, to obtain a generalized counting oracle we need access to a minimum KL-divergence oracle.  However, later in this section we show that given a generalized counting oracle, we can not only solve  the max-entropy convex programs as in Theorem~\ref{thm:approximate-count}, but also the min KL-divergence program in a straightforward manner. The convex program for the min-KL-divergence problem is given in Figure~\ref{conv-KLa}. Given a $\mu\in \RR^m$, recall that $p^{\mu}$ is the product distribution $p^{\mu}_M \defeq \frac{e^{-\mu(M)}}{Z^{\mu}}.$ 
\begin{figure}[t]
\begin{center}
\begin{eqnarray*}
  \sup & \sum_{M\in \M} p_M \ln\frac{p^{\mu}(M)}{p_M} + \ln Z^{\mu}  \\
\st  \\
\forall e\in [m]& \sum_{M\in \M, \; M \ni e} p_M=\theta_e \\
  & \sum_{M\in \M} p_M =1\\
\forall M\in \M & p_M \geq 0
\end{eqnarray*}
\end{center}
\caption{Min-KL-Divergence Program for ($\M,\vec{\theta},\mu$)} \label{conv-KLa}
\end{figure}
Observe that the objective is to find a distribution $p$ that \emph{minimizes} the KL-divergence, up to a shift, between the distributions $p$ and $p^{\mu}$. This follows since the objective can be rewritten as $$\sum_{M\in \M} p_M \ln\frac{p^{\mu}(M)}{p_M} + \ln Z^{\mu}= \ln Z^{\mu}-D_{KL}(p\|p^{\mu})$$ where $Z^{\mu}$ does not depend on the variable $p$ but only on the input $\mu$. The dual of this convex program is given in Figure~\ref{conv-dual-KLa}. 
We use the following to denote the objective function of the dual:
$$ f^{\mu}_\theta(\vec{\lambda})\defeq \sum_{e\in [m]}\theta_e \lambda_e+\ln \sum_{M\in \M}e^{-\lambda(M)-\mu(M)}+\ln Z^{\mu}.$$
\begin{figure}[t]
\begin{center}
\begin{eqnarray*}
\inf &  \sum_{e\in [m]}\theta_e \lambda_e+\ln \sum_{M\in \M}e^{-\lambda(M)-\mu(M)}+\ln Z^{\mu} \\
\st \\
\forall e\in [m] & \lambda_e \in \mathbb{R}
\end{eqnarray*}
\end{center}
\caption{Dual of the Min-KL-Divergence Program for ($\M,\vec{\theta},\mu$)} \label{conv-dual-KLa}
\end{figure}
When $\theta$ is in the interior of $P(\M),$ strong duality holds between the programs of Figure \ref{conv-KLa} and \ref{conv-dual-KLa}.
We assume that we are given the following approximate oracle to solve the above set of convex programs. 
\begin{definition}\label{def:approxKLoracle} An {\em approximate KL-optimization oracle} for $\M,$ given a  $\vtheta$ in the $\eta$-interior of $P(\M),$ $\mu\in \RR^{m},$ a $\zeta >0$, and an $ \epsilon >0$, either
\begin{enumerate}
\item asserts that $\inf_{\vlambda}f^{\mu}_{\vtheta}(\vlambda)\geq \zeta -\epsilon$ or
\item  returns a $\vlambda\in \RR^{m}$ such that $f^{\mu}_{\vtheta}(\vlambda)\leq \zeta+\epsilon$.
\end{enumerate}
The oracle is assumed to be efficient, i.e., it runs in time polynomial in $m$, $\nfrac{1}{\epsilon},$ $\nfrac{1}{\eta}$, the number of bits needed to represent $\zeta$ and  $\|\mu\|_1.$ 
\end{definition}
\noindent
The following theorem is the appropriate generalization of Theorem~\ref{thm:reverse} in this setting.

\begin{theorem}\label{thm:reverse2}
There exists an algorithm that,  given a maximal set of linearly independent equalities $(A^=,b)$ and a separation oracle for $P(\M),$ a $\mu\in \RR^{m},$ and an  approximate KL-optimization oracle for  $\M$ as above,  returns a $\widetilde{Z}$ such that $(1-\epsilon)Z^{\mu}\leq \widetilde{Z}\leq (1+\epsilon)Z^{\mu}.$ Assuming that the running times of the separation oracle and the approximate KL oracle are polynomial in their respective  input parameters, the  running time of the algorithm is bounded by a polynomial in $m,$ $\nfrac{1}{\epsilon}$, the number of bits needed to represent $(A^=,b)$ and $\|\mu\|_1.$
\end{theorem}

\noindent
We omit the proof since it is a simple but tedious generalization of the proof of Theorem~\ref{thm:reverse}. We highlight below the key additional points that must be taken into account. The algorithm in the proof of Theorem \ref{thm:reverse2} is obtained by using the ellipsoid algorithm to maximize the  concave function

$$\max_{\theta} g^{\mu}(\theta)$$
over the interior of $P(\M).$  Here,
$g^{\mu}(\theta)\defeq \min_{\lambda} f^{\mu}_{\theta}(\lambda).$
Indeed, the maximum is attained at $\theta^{\star}$ where $$\theta^{\star}_e=\frac{\sum_{M:e\in M}e^{-\mu(M)}}{Z^{\mu}}$$ and the objective value at this maximum is $\ln Z^{\mu}$. Thus, we use the ellipsoid algorithm to search for $\theta^{\star}$. The issue of interiority as in Lemma~\ref{lem:interior-count} is resolved by proving the following lemma.

\begin{lemma}\label{lem:interior-count2}
Given an $\epsilon>0$, there exists an $\eta>0$ and $\thetap$ such that ${\thetap}$ is in $\eta$-interior of $P(\M)$ and $g^{\mu}(\thetap)\geq \left(1-\frac{\epsilon}{16m\|\mu\|_1}\right)\ln Z^{\mu} \geq \ln Z^{\mu} -\frac{\epsilon}{16}$. Moreover, $\eta$ is at least a polynomial in $\nfrac{1}{m}$, $\eps$ and $\nfrac{1}{\|\mu\|_1}.$
\end{lemma}

\noindent 
Similarly, Lemma~\ref{lem:marginal_entropy} can be generalized to show the following:

\begin{lemma}\label{lem:marginal_entropy2}
Let $\vec{\theta},\vec{\theta}'\in P(\M)$ such that $\|\vec{\theta}-\vec{\theta}'\|\leq \epsilon$ and $\vec{\theta}$ is in $\eta$-interior of $P(\M)$. Then
$$g^{\mu}({\theta}')  \geq \left(1-\frac{\epsilon}{\eta \|\mu\|_1}\right)g(\theta).$$
\end{lemma}

\noindent
Using the above two lemmas it is straightforward to generalize the argument in Theorem~\ref{thm:reverse} to prove Theorem~\ref{thm:reverse2}.

To complete the picture we show that, given access to a generalized counting oracle, we can solve the above pair of convex programs. This gives the following theorem which is a generalization of Theorem \ref{thm:approximate-count}. 

\begin{theorem}\label{thm:approximate-count1}
There exists an algorithm that, given a maximal set of linearly independent equalities $(A^=,b)$ and a  generalized approximate counting oracle for $P(\M) \subseteq \mathbb{R}^m,$   a $\vec{\theta}$ in the $\eta$-interior of $P(\M)$, $\mu\in \RR^m$ and an $\epsilon>0,$ returns a $\vec{\lambda}^\circ$ such that
$$f^{\mu}_\theta(\vec{\lambda}^\circ) \leq f^{\mu}_\theta(\vec{\lambda}^\star) + \eps,$$
where $\vec{\lambda}^{\star}$ is the optimal solution to the dual of the max-entropy convex program for $(\M,\vec{\theta},\mu)$ from Figure \ref{conv-dual-KLa}. Assuming that the generalized approximate counting oracle is polynomial in its input parameters, the running time of the algorithm is polynomial in $m$, $\nfrac{1}{\eta},$ $\log \nfrac{1}{\epsilon}$, the number of bits needed to represent  $\theta$ and  $(A^=,b),$ and  $\|\mu\|_1.$ 
\end{theorem}

\noindent
A similar theorem can be stated in the exact counting oracle setting, extending Theorem \ref{thm:strong}.
The proof of Theorem \ref{thm:approximate-count1} is quite straightforward and relies on the following lemma that states that the objective of the primal convex program is just an additive shift from the objective of the maximum entropy convex program. Thus, the primal optimum solution remains the same. This implies that the dual convex program can be solved as in the proof of Theorem \ref{thm:approximate-count} with one additional call to the generalized approximate counting oracle involving $\mu$s. 
\begin{lemma}
Let $p$ be any feasible solution to the primal convex program given in Figure~\ref{conv-KLa}. Then the objective
$$ \sum_{M\in \M}p_M \ln\frac{p^{\mu}_M}{p_M} + \ln Z^{\mu}= \sum_{M\in \M}p_M \ln\frac{1}{p_M} - \langle \theta, \mu\rangle.$$
\end{lemma}
\begin{proof}
We have
\begin{eqnarray*}
\sum_{M\in \M}p_M \ln\frac{p^{\mu}_M}{p_M} + \ln Z^{\mu} &=&\sum_{M\in \M}p_M \ln\frac{1}{p_M} +  \sum_{M\in \M}p_M \ln p^{\mu}_M+  \ln Z^{\mu}\\
&=&\sum_{M\in \M}p_M \ln\frac{1}{p_M} +  \sum_{M\in \M}p_M \ln \frac{e^{-\mu(M)}}{Z^{\mu}}+  \ln Z^{\mu}\\
&=&\sum_{M\in \M}p_M \ln\frac{1}{p_M} + \sum_{M\in \M}p_M (-\mu(M))+  \sum_{M\in \M}p_M Z^{\mu}+  \ln Z^{\mu}\\
&=&\sum_{M\in \M}p_M \ln\frac{1}{p_M} -\sum_{e\in E}\mu_e \sum_{M:e\in M}p_M \\
&=&\sum_{M\in \M}p_M \ln\frac{1}{p_M} - \sum_{e\in E}\mu_e \theta_e
\end{eqnarray*}
where we have used the facts that $\sum_{M}p_M=1$ and $\sum_{M:e\in M} p_M=\theta_e$ since $p$ satisfies the constraints of the convex program in Figure~\ref{conv-KLa}.
\end{proof}

\section{$2\rightarrow 2$-norm of $\nabla f$}\label{sec:grad}
For a fixed $\theta,$ let 
\begin{equation}\label{eq:f}
f(\lambda) \defeq \langle \theta, \lambda \rangle + \ln \sum_{M \in \M} e^{-\lambda(M)}.
\end{equation}
Recall that $ \|\nabla f \|_{2 \rightarrow 2}$ is the $2 \rightarrow 2$ Lipschitz constant of $\nabla f$ and is defined to be the smallest non-negative number such that
$$ \left\| \nabla f(\lambda)-\nabla f(\lambda')  \right\|_2 \leq  \|\nabla f \|_{2 \rightarrow 2} \cdot \|\lambda-\lambda'\|_2$$
for all $\lambda,\lambda'.$
In this section, we show the following theorem, which can be used to give alternative gradient-descent based proofs of Theorems \ref{thm:strong} and \ref{thm:approximate-count}.

\begin{theorem}\label{thm:2to2}
Let $f$ be defined as in \eqref{eq:f}. Then,   $ \|\nabla f \|_{2 \rightarrow 2} \leq O(m\sqrt{m}).$ 
\end{theorem}

\begin{proof}
Given $\lambda_1$ and $\lambda_2$, let $p^{\lambda_1}$ and $p^{\lambda_2}$ be the corresponding product distributions and let $\theta_1$ and $\theta_2$ be the corresponding marginals. 
We break  the calculation of $\|\nabla f \|_{2 \rightarrow 2}$  into two parts: 
$$\|\lambda_1-\lambda_2\|_2 > \frac{1}{10 \sqrt{m}} \ \ {\rm and} \ \  \|\lambda_1-\lambda_2\|_2 \leq \frac{1}{10 \sqrt{m}}.$$
Estimating  $\|\nabla f \|_{2 \rightarrow 2}$ in the first case is straightforward. To see this, recall that 
\begin{equation}\label{eq:22nabla}
\nabla f(\lambda_i)=\theta-\theta_i
\end{equation}
 for $i=1,2.$ Thus, 
\begin{equation}\label{eq:22grad}
 \|\nabla f(\lambda_i) \|_2 \leq \sqrt {m}
\end{equation}
since $\theta,\theta_i \in P(\M) \subseteq [0,1]^m$, which implies that  $\theta-\theta_i \in [-1,1]^m$ for $i=1,2.$ 
Hence, $$\|\nabla f (\lambda_1) - \nabla f(\lambda_2)\|_2 \leq 2\sqrt{m}.$$
Thus, when   
\begin{equation}\label{eq:22assume1}
\|\lambda_1-\lambda_2\|_2 > \frac{1}{10 \sqrt{m}},
\end{equation} 
we obtain 
$$ \|\nabla f (\lambda_1) - \nabla f(\lambda_2)\|_2 \stackrel{\Delta{\rm -ineq.}}{\leq}  \|\nabla f (\lambda_1) \|_2 + \| \nabla f(\lambda_2)\|_2   \stackrel{\eqref{eq:22grad}}{\leq} 2 \sqrt{m} = 20 m \cdot  \frac{1}{10 \sqrt{m}} \stackrel{\eqref{eq:22assume1}}{<} 20 m \cdot \|\lambda_1-\lambda_2\|_2.$$ 
Hence, we move on to proving the theorem in the case
\begin{equation}\label{eq:22assume}
 \|\lambda_1-\lambda_2\|_2 \leq \frac{1}{10 \sqrt{m}}
\end{equation}
Towards this, define 
\begin{equation}\label{eq:22eps}
 \eps \defeq \sqrt{m} \|\lambda_1-\lambda_2\|_2.
\end{equation}
Note that by assumption \eqref{eq:22assume}, $\eps \leq \nfrac{1}{10}.$
It follows from \eqref{eq:22eps} that for any  $M\in \M$, 
\begin{equation}\label{eq:lambdaM}
\vert \lambda_1(M)-\lambda_2(M)\vert = \left\vert \sum_{e \in M} \left( (\lambda_1)_e - (\lambda_2)_e \right) \right\vert \stackrel{\Delta{\rm -ineq.}}{\leq}   \sum_{e \in M} \left\vert (\lambda_1)_e - (\lambda_2)_e\right\vert  \leq   \|\lambda_1-\lambda_2\|_1 \stackrel{{\rm Cau.-Sch.}}{\leq} \sqrt{m}\|\lambda_1-\lambda_2\|_2 = \epsilon.
\end{equation}
The following series of claims establishes Theorem \ref{thm:2to2}.  

\begin{claim}\label{clm:22-1}
$e^{-\epsilon}\leq \frac{Z^{\lambda_1}}{Z^{\lambda_2}} \leq e^{\epsilon}.$
\end{claim}

\begin{proof}
For each $M \in \M,$  \eqref{eq:lambdaM} implies that
\begin{equation}\label{eq:lambdaexpM}
e^{-\epsilon}\leq \frac{e^{-\lambda_1(M)}}{e^{-\lambda_2(M)}}\leq e^{\epsilon}.
\end{equation}
Thus,
\begin{equation}\label{eq:22eps1}
e^{-\epsilon}\leq  \min_{M\in \M}\frac{e^{-\lambda_1(M)}}{e^{-\lambda_2(M)}}\leq \frac{\sum_{M\in \M} e^{-\lambda_1(M)}}{\sum_{M\in \M}e^{-\lambda_2(M)}} \leq \max_{M\in \M}  \frac{e^{-\lambda_1(M)}}{e^{-\lambda_2(M)}}\leq  e^{\epsilon}.
\end{equation}
Here, we have used the inequality that for non-negative numbers $a_1,a_2,\ldots$ and $b_1,b_2,\ldots,$ 
\begin{equation}\label{eq:min-max}
 \min_{i} \frac{a_i}{b_i} \leq \frac{\sum_i a_i}{\sum_i b_i} \leq  \max_{i} \frac{a_i}{b_i}.
\end{equation}
The claim follows by combining \eqref{eq:22eps1} with the definition 
$$ \frac{\sum_{M\in \M} e^{-\lambda_1(M)}}{\sum_{M\in \M}e^{-\lambda_2(M)}} = \frac{Z^{\lambda_1}}{Z^{\lambda_2}}.$$
\end{proof}

\begin{claim}\label{clm:22-2}
For each $M\in \M$, $e^{-2\epsilon}\leq \frac{p^{\lambda_1}_M}{p^{\lambda_2}_M}\leq e^{2\epsilon}.$
\end{claim}
\begin{proof}
By definition,
$$\frac{p^{\lambda_1}_M}{p^{\lambda_2}_M}=\frac{e^{-\lambda_1(M)}}{e^{-\lambda_2(M)}}\cdot \frac{Z^{\lambda_2}}{Z^{\lambda_1}}.$$
Since all the numbers involved in this product  are positive, \eqref{eq:lambdaexpM} and Claim \ref{clm:22-1} imply that 
both the ratios in the right hand side of the equation are bounded from below by $e^{-\epsilon}$ and from above by $e^\eps.$ Hence, their product is bounded from below by $e^{-2\epsilon}$ and from above by $e^{2 \eps}$, completing the proof of the claim.
\end{proof}

\begin{claim}\label{cl:theta-e}
For each $e\in [m]$, $e^{-2\epsilon}\leq \frac{(\theta_1)_e}{(\theta_2)_e}\leq e^{2\epsilon}.$
\end{claim}
\begin{proof}
By the definition of $\theta_1$ and $\theta_2,$ 
$$\frac{(\theta_1)_e}{(\theta_2)_e}=\frac{\sum_{M \in \M, M \ni e} p^{\lambda_1}_M}{\sum_{M \in \M, M \ni e}p^{\lambda_2}_M}.$$
Combining Claim \ref{clm:22-2} and \eqref{eq:min-max}, we obtain that
$$e^{-2\epsilon}\leq \min_{M \in \M, M \ni e} \frac{ p^{\lambda_1}_M}{p^{\lambda_2}_M} \leq \frac{\sum_{M \in \M, M \ni e} p^{\lambda_1}_M}{\sum_{M \in \M, M \ni e}p^{\lambda_2}_M} \leq \max _{M \in \M, M \ni e}\frac{ p^{\lambda_1}_M}{p^{\lambda_2}_M}\leq e^{2\epsilon},$$
completing the proof of the claim.
\end{proof}

\begin{claim}
For $\eps \leq \nfrac{1}{10},$  $\|\theta_1-\theta_2\|_1\leq 3\epsilon m.$
\end{claim}
\begin{proof}
By definition
$$
\|\theta_1-\theta_2\|_1= \sum_{e\in [m]} \vert (\theta_1)_e-(\theta_2)_e\vert
$$
Since $\theta_1,\theta_2 \geq 0,$ 
 Claim~\ref{cl:theta-e} implies that for each $e \in [m],$ 
$$(e^{-2\epsilon}-1)(\theta_2)_e \leq (\theta_1)_e-(\theta_2)_e\leq (e^{2\epsilon}-1)(\theta_2)_e.$$
Since $\max\{\vert e^{-2\epsilon}-1\vert, \vert e^{2\epsilon}-1\vert\}\leq 3\epsilon $ for $\epsilon \leq \nfrac{1}{10},$ the above inequality reduces to, for each $e \in [m],$  
$$\vert (\theta_1)_e-(\theta_2)_e\vert \leq 3\epsilon (\theta_2)_e.$$
Thus, 
$$\|\theta_1-\theta_2\|_1=\sum_{e\in [m]}  \vert (\theta_1)_e-(\theta_2)_e\vert \leq \sum_{e\in [m]} 3\epsilon (\theta_2)_e \leq 3\epsilon m$$
since $\theta_2 \geq 0$ and $\theta_2 \in P(\M) \subseteq [0,1]^m.$ 
This completes the proof of the claim.
\end{proof}
\noindent
Finally, to complete the proof of Theorem \ref{thm:2to2}, note that when \eqref{eq:22assume} holds,
$$\| \nabla f(\lambda_1)-\nabla f(\lambda_2)  \|_2 \stackrel{\eqref{eq:22nabla}}{=} \|\theta_1-\theta_2\|_2 \stackrel{\ell_2(x,y) \leq \ell_1(x,y)}{\leq} \|\theta_1-\theta_2\|_1 \stackrel{{\rm Claim  \; \ref{cl:theta-e}}}{\leq}  3m \epsilon \stackrel{\eqref{eq:22eps}}{=} 3 m \sqrt{m} \|\lambda_1-\lambda_2\|_2.$$
\end{proof}

\end{document}